%% file: arxiv.tex
\documentclass[11pt]{article}
\usepackage{hyperref}       
\usepackage{url}            
\usepackage{booktabs}       
\usepackage{amsfonts}       
\usepackage{nicefrac}       
\usepackage{microtype}      
\usepackage{framed}
\usepackage{fullpage}
\usepackage{amsmath,amssymb,amsthm,amsfonts,comment,graphicx,subcaption}
\usepackage{color}
\usepackage{float}
\usepackage{booktabs}
\theoremstyle{plain}
\newtheorem{theorem}{Theorem}
\newtheorem{question}{Problem}

\newtheorem{assumption}{Assumption}
\newtheorem{lemma}[theorem]{Lemma}

\newtheorem{claim}{Claim}
\newtheorem{corollary}[theorem]{Corollary}
\theoremstyle{definition}
\newtheorem{definition}{Definition}

\newcommand{\wh}{\widehat}
\newcommand{\wt}{\widetilde}
\newcommand{\ov}{\overline}

\newcommand{\R}{\mathbb{R}}
\renewcommand{\epsilon}{\varepsilon}
\renewcommand{\tilde}{\wt}
\renewcommand{\hat}{\wh}
\newcommand{\loss}{\mathcal{L}}

\DeclareMathOperator*{\E}{{\bf {E}}}
\DeclareMathOperator*{\var}{\mathrm{Var}}

\DeclareMathOperator{\argmin}{argmin}
\DeclareMathOperator{\argmax}{argmax}
\DeclareMathOperator{\poly}{poly}
\DeclareMathOperator{\polylog}{polylog}
\DeclareMathOperator{\nnz}{nnz}

\DeclareMathOperator{\mc}{mc}
\DeclareMathOperator{\mmc}{mmc}
\DeclareMathOperator{\mma}{mma}
\DeclareMathOperator{\diag}{diag}
\newcommand{\Ab}{\overline{A}}

\input{libtheorems.tex}
\title{Efficient Symmetric Norm Regression via Linear Sketching}
\author{Zhao Song \\
University of Washington \\
\texttt{magic.linuxkde@gmail.com}
\and Ruosong Wang \\
Carnegie Mellon University \\
\texttt{ruosongw@andrew.cmu.edu}
\and Lin F. Yang \\
University of California, Los Angeles \\
\texttt{linyang@ee.ucla.edu}
\and Hongyang Zhang \\
Toyota Technological Institute at Chicago \\
\texttt{hongyanz@ttic.edu}
\and Peilin Zhong \\
Columbia University \\
\texttt{pz2225@columbia.edu}
}
\date{}
\begin{document}

\maketitle

\input{abstract}

\input{intro}

\input{orlicz_nips}
\input{symmetric_nips}

\input{conclusion}
\input{ack}
\bibliography{bib}
\appendix
\input{pre}
\input{orlicz_appendix}

\input{symmetric_appendix}

\input{main_appendix}

\input{exp}
\end{document}

%% file: libtheorems.tex
\usepackage{etoolbox}

\makeatletter

\newcommand{\define}[4][ignore]{%
  \ifstrequal{#1}{ignore}{}{
  \@namedef{thmtitle@#2}{#1}}%
  \@namedef{thm@#2}{#4}%
  \@namedef{thmtypen@#2}{lemma}%
  \newtheorem{thmtype@#2}[theorem]{#3}%
  \newtheorem*{thmtypealt@#2}{#3~\ref{#2}}%
}

\newcommand{\state}[1]{%
  \@namedef{curthm}{#1}
  \@ifundefined{thmtitle@#1}{
  \begin{thmtype@#1}
    }{
  \begin{thmtype@#1}[\@nameuse{thmtitle@#1}]
  }
    \label{#1}
    \@nameuse{thm@#1}
  \end{thmtype@#1}
  \@ifundefined{thmdone@#1}{
  \@namedef{thmdone@#1}{stated}%
  }{}
}

\newcommand{\restate}[1]{%
  \@namedef{curthm}{#1}
  \@ifundefined{thmtitle@#1}{
    \begin{thmtypealt@#1}
    }{
  \begin{thmtypealt@#1}[\@nameuse{thmtitle@#1}]
  }
    \@nameuse{thm@#1}
  \end{thmtypealt@#1}
  \@ifundefined{thmdone@#1}{
  \@namedef{thmdone@#1}{stated}%
  }{}
}

\newcommand{\thmlabel}[1]{
  \@ifundefined{thmdone@\@nameuse{curthm}}{\label{#1}
    }{\tag*{\eqref{#1}}}
}
\makeatother

%% file: abstract.tex
\begin{abstract}
We provide efficient algorithms for overconstrained linear regression problems with size $n \times d$ when the loss function is a symmetric norm (a norm invariant under sign-flips and coordinate-permutations). An important class of symmetric norms are Orlicz norms, where for a function  $G$ and a vector $y \in \mathbb{R}^n$, the corresponding Orlicz norm $\|y\|_G$ is defined as the unique value $\alpha$ such that $\sum_{i=1}^n G(|y_i|/\alpha) = 1$. When the loss function is an Orlicz norm, our algorithm produces a $(1 + \varepsilon)$-approximate solution for an arbitrarily small constant $\varepsilon > 0$ in input-sparsity time, improving over the previously best-known algorithm which produces a $d \cdot \polylog n$-approximate solution. When the loss function is a general symmetric norm, our algorithm produces a $\sqrt{d} \cdot \polylog n \cdot \mathrm{mmc}(\ell)$-approximate solution in input-sparsity time, where $\mathrm{mmc}(\ell)$ is a quantity related to the symmetric norm under consideration. To the best of our knowledge, this is the first input-sparsity time algorithm with provable guarantees for the general class of symmetric norm regression problem. Our results shed light on resolving the universal sketching problem for linear regression, and the techniques might be of independent interest to numerical linear algebra problems more broadly.

\end{abstract}

%% file: intro.tex
\section{Introduction}
\label{sec:intro}
Linear regression is a fundamental problem in machine learning.
For a data matrix $A \in \mathbb{R}^{n \times d}$ and a response vector $b \in \mathbb{R}^n$ with $n \gg d$, the overconstrained linear regression problem can be formulated as solving the following optimization problem:
\begin{align}\label{eqn:regress}
\min_{x\in \mathbb{R}^d} \loss(Ax - b), 
\end{align}
where $\loss : \mathbb{R}^n \to \mathbb{R}$ is a loss function.
Via the technique of linear sketching, we have witnessed many remarkable speedups for linear regression for a wide range of loss functions.
Such technique involves designing a sketching matrix $S \in \mathbb{R}^{r \times n}$, and showing that by solving a linear regression instance on the data matrix $SA$ and the response vector $Sb$, which is usually much smaller in size, one can obtain an approximate solution to the original linear regression instance in~\eqref{eqn:regress}.
Sarl\'os showed in \cite{sarlos2006improved} that by taking $S$ as a Fast Johnson-Lindenstrauss Transform matrix \cite{ailon2006approximate},  one can obtain $(1 + \varepsilon)$-approximate solutions to the least square regression problem ($\loss(y) = \|y\|_2^2$) in $O(nd \log n + \poly(d / \varepsilon))$ time.
The running time was later improved to  $O(\nnz(A) + \poly(d / \varepsilon))$ \cite{clarkson2013low, meng2013low, nelson2013sparsity, li2013iterative, cohen2016nearly}. 
Here $\nnz(A)$ is the number of non-zero entries in the data matrix $A$, which could be much smaller than $nd$ for sparse data matrices. 
This technique was later generalized to other loss functions. 
By now, we have $\widetilde{O}(\nnz(A) + \poly(d / \varepsilon))$ 
time algorithms for $\ell_p$ norms ($\loss(y) = \|y\|_p^p$)~\cite{dasgupta2009sampling, meng2013low, woodruff2013subspace, cohen2015p, wang2019tight}, the quantile loss function \cite{yang2013quantile}, $M$-estimators \cite{clarkson2015sketching, cw15b} and the Tukey loss function~\cite{clarkson2019dimensionality}.

Despite we have successfully applied the technique of linear sketching to many different loss functions, ideally, it would be more desirable to design algorithms that work for a wide range of loss functions, instead of designing a new sketching algorithm for every specific loss function.
Naturally, this leads to the following problem, which is the linear regression version of the universal sketching problem\footnote{\url{https://sublinear.info/index.php?title=Open_Problems:30}.} studied in streaming algorithms \cite{braverman2010zero, braverman2016streaming}.
We note that similar problems are also asked and studied for various algorithmic tasks, including principal component analysis \cite{song2018towards}, sampling \cite{lsv18}, approximate nearest neighbor search~\cite{andoni2017approximate, andoni2018holder}, discrepancy \cite{ndtt18,b18}, sparse recovery~\cite{nsw19b} and mean estimation with statistical queries~\cite{feldman2017statistical, li2019mean}.

\begin{question} \label{question:universal}
Is it possible to design sketching algorithms for linear regression, that work for a wide range of loss functions?
\end{question}

\begin{table}[t]
\caption{$M$-estimators}
\label{tab:M-estimators}
\begin{center}
\begin{small}
\begin{sc}
\begin{tabular}{lcc}
    \toprule
      \small{Huber} & \small{$\left\{\begin{array}{ll}x^2/2&|x|\leq c\\ c(|x|-c/2) & |x|>c\end{array}\right.$} \\ \hline
      \small{$\ell_1-\ell_2$} &  \small{$2(\sqrt{1+x^2/2}-1)$}   \\ \hline
      \small{``Fair"} & \small{$c^2\left(|x|/c-\log(1+|x|/c)\right)$}\\
\bottomrule
\end{tabular}
\end{sc}
\end{small}
\end{center}
\vspace{-0.2in}

\end{table}

Prior to our work, \cite{clarkson2015sketching, cw15b} studied this problem in terms of $M$-estimators, where the loss function employs the form $\loss(y)=\sum_{i=1}^n G(y_i)$ for some function $G$.
See Table~\ref{tab:M-estimators} for a list of $M$-estimators.
However, much less is known for the case where the loss function $\loss(\cdot)$ is a norm, except for $\ell_p$ norms.
Recently, Andoni et al.~\cite{alszz18} tackle Problem~\ref{question:universal} for Orlicz norms, which can be seen as a scale-invariant version of $M$-estimators.
For a function  $G$ and a vector $y \in \mathbb{R}^n$ with $y \neq 0$, the corresponding Orlicz norm $\|y\|_G$ is defined as the unique value $\alpha$ such that 
\begin{equation}\label{equ:orlicz}
\sum_{i=1}^n G(|y_i|/\alpha) = 1.
\end{equation}
When $y = 0$, we define $\|y\|_{G}$ to be $0$.
Note that Orlicz norms include $\ell_p$ norms as special cases, by taking $G(z)=|z|^p$ for some $p \ge 1$. 
Under certain assumptions on the function $G$, \cite{alszz18} obtains the first input-sparsity time algorithm for solving Orlicz norm regression.
More precisely, in $\widetilde{O}(\nnz(A) + \poly(d \log n))$ time, their algorithm obtains a solution $\hat{x}\in \mathbb{R}^{d}$ such that
$
\|A\hat{x} - b\|_G\le d \cdot \polylog n \cdot \min_{x\in \mathbb{R}^d}\|A{x} - b\|_G
$.

There are two natural problems left open by the work of \cite{alszz18}.
First, the algorithm in \cite{alszz18} has approximation ratio as large as $d \cdot \polylog n$.
Although this result is interesting from a theoretical point of view, such a large approximation ratio is prohibitive for machine learning applications in practice.
Is it possible to obtain an algorithm that runs in $\widetilde{O}(\nnz(A) + \poly(d/ \varepsilon))$ time, with approximation ratio $1 + \varepsilon$, for arbitrarily small $\varepsilon$, similar to the case of $\ell_p$ norms?
Moreover, although Orlicz norm includes a wide range of norms, many other important norms, e.g., 
top-$k$ norms (the sum of absolute values of the leading $k$ coordinates of a vector), 
 max-mix of $\ell_p$ norms (e.g. $\max\{\|x\|_2, c\|x\|_1\}$ for some $c > 0$),
and sum-mix of $\ell_p$ norms (e.g. $\|x\|_2+c\|x\|_1$ for some $c > 0$),
are not Orlicz norms.
More complicated examples include the $k$-support norm \cite{argyriou2012sparse} and the box-norm \cite{mcdonald2014spectral}, which have found applications in sparse recovery.
In light of Problem~\ref{question:universal}, it is natural to ask whether it is possible to apply the technique of linear sketching to a broader class of norms. 
In this paper, we obtain affirmative answers to both problems, and make progress towards finally resolving Problem~\ref{question:universal}.

\paragraph{Notations.}
We use $\widetilde{O}(f)$ to denote $f  \polylog f $. 
For a matrix $A \in \mathbb{R}^{n \times d}$, we use $A_{i} \in \mathbb{R}^d$ to denote its $i$-th row, viewed as a column vector.
For $n$ real numbers $x_1, x_2,\ldots, x_n$, we define $\diag(x_1,x_2,\ldots,x_n) \in \R^{n\times n}$ to be the diagonal matrix where the $i$-th diagonal entry is $x_i$.
For a vector $x \in \R^n$ and $p \ge 1$, we use $\|x\|_p$ to denote its $\ell_p$ norm, and $\|x\|_0$ to denote its $\ell_0$ norm, i.e., the number of non-zero entries in $x$.
For two vectors $x, y \in \R^n$, we use $\langle x, y \rangle$ to denote their inner product.
For any $n > 0$, we use $[n]$ to denote the set $\{1,2,\ldots,n\}$.
For $0 \le p \le 1$, we define $\mathrm{Ber}(p)$ to be the Bernoulli distribution with parameter $p$.
We use $\mathbb{S}^{n - 1}$ to denote the unit $\ell_2$ sphere in $\mathbb{R}^n$, i.e., $\mathbb{S}^{n - 1} = \{x \in \mathbb{R}^n \mid \|x\|_2 = 1\}$.
We use $\mathbb{R}_{\ge 0}$ to denote the set of all non-negative real numbers, i.e., $\mathbb{R}_{\ge 0} = \{x \in \mathbb{R} \mid x \ge 0\}$.
\subsection{Our Contributions}
\paragraph{Algorithm for Orlicz Norms.}Our first contribution is a unified algorithm which produces $(1+ \varepsilon)$-approximate solutions to the linear regression problem in \eqref{eqn:regress}, when the loss function $\loss(\cdot)$ is an Orlicz norm. 
Before introducing our results, we first give our assumptions on the function $G$ which appeared in \eqref{equ:orlicz}.

\begin{assumption}\label{assump:property_P}
We assume the function $G:\mathbb{R}\to\mathbb{R}_{\ge 0}$ satisfies the following properties:
\begin{enumerate}
\item $G$ is a strictly increasing convex function on $[0,\infty)$;
\item $G(0)=0$, and for all $x\in\mathbb{R}$, $G(x)=G(-x)$;
\item There exists some $C_G>0$, such that for all $0<x<y$, $G(y)/G(x)\leq C_G(y/x)^2$.
\end{enumerate}
\end{assumption}

The first two conditions in Assumption~\ref{assump:property_P} are necessary to make sure the corresponding Orlicz norm $\| \cdot \|_G$ is indeed a norm, and the third condition requires the function $G$ to have at most quadratic growth, which can be satisfied by all $M$-estimators in Table~\ref{tab:M-estimators} and is also required by prior work~\cite{alszz18}.
Notice that our assumptions are weaker than those in~\cite{alszz18}.
In~\cite{alszz18}, it is further required that $G(x)$ is a linear function when $x > 1$, and $G$ is twice differentiable on an interval $(0, \delta_G)$ for some $\delta_G > 0$.
Given our assumptions on $G$, our main theorem is summarized as follows. 

\begin{theorem}\label{thm:orlicz}
For a function $G$ that satisfies Assumption \ref{assump:property_P}, there exists an algorithm that, on any input $A\in \mathbb{R}^{n\times d}$ and $b\in \mathbb{R}^n$, finds a vector $x^*$ in time $\wt{O}(\nnz(A) + \poly(d / \varepsilon))$, such that with probability at least $0.9$,
$
\|Ax^* - b\|_G\le (1+ \varepsilon) \min_{x\in \mathbb{R}^d}\|A{x} - b\|_G
$.
\end{theorem}
To the best of our knowledge, this is the first input-sparsity time algorithm with $(1+ \varepsilon)$-approximation guarantee, that goes beyond $\ell_p$ norms, the quantile loss function, and $M$-estimators.
See Table~\ref{table: comparisons of subspace embedding} for a more comprehensive comparison with previous results.

\paragraph{Algorithm for Symmetric Norms.}We further study the case when the loss function $\loss(\cdot)$ is a symmetric norm.
Symmetric norm is a more general class of norms, which includes all norms that are invariant under sign-flips and coordinate-permutations. 
Formally, we define symmetric norms as follow. 
\begin{definition}
A norm $\|\cdot\|_{\ell}$ is called a {\em symmetric norm}, if $\|(y_1, y_2, \ldots, y_n)\|_{\ell} = \|(s_1y_{\sigma_1}, s_2y_{\sigma_2}, \ldots, s_ny_{\sigma_n})\|_{\ell}$ for any permutation $\sigma$ and any assignment of $s_i\in \{-1, 1\}$.
\end{definition}

Symmetric norm includes $\ell_p$ norms and Orlicz norms as special cases.
It also includes all examples provided in the introduction, i.e., top-$k$ norms, max-mix of $\ell_p$ norms, sum-mix of $\ell_p$ norms, the $k$-support norm \cite{argyriou2012sparse} and the box-norm \cite{mcdonald2014spectral}, as special cases.
Understanding this general set of loss functions can be seen as a preliminary step to resolve Problem~\ref{question:universal}.
Our main result for symmetric norm regression is summarized in the following theorem.
\begin{theorem}\label{thm:sym}
Given a symmetric norm $\|\cdot\|_{\ell}$,  there exists an algorithm that, on any input $A\in \mathbb{R}^{n\times d}$ and $b\in \mathbb{R}^n$, finds a vector $x^*$ in time $\wt{O}(\nnz(A) + \poly(d))$, such that with probability at least $0.9$,
	$
	\|Ax^* - b\|_{\ell}\le  \sqrt{d} \cdot \polylog n\cdot  \mmc(\ell) \cdot\min_{x\in \mathbb{R}^d}\|A{x} - b\|_{\ell}
	$.
\end{theorem}
In the above theorem, $\mmc(\ell)$ is a characteristic of the symmetric norm $\|\cdot\|_\ell$, which has been proven to be essential in streaming algorithms for symmetric norms~\cite{blasiok2017streaming}. See Definition~\ref{def:mmc} for the formal definition of $\mmc(\ell)$, and Section~\ref{sec:alg} for more details about $\mmc(\ell)$. 
In particular, for $\ell_p$ norms with $p \le 2$, top-$k$ norms with $k \ge n / \polylog n$,  max-mix of $\ell_2$ norm and $\ell_1$ norm ($\max\{\|x\|_2, c\|x\|_1\}$ for some $c > 0$),
sum-mix of $\ell_2$ norm and $\ell_1$ norm ($\|x\|_2+c\|x\|_1$ for some $c > 0$),
the $k$-support norm, and the box-norm, $\mmc(\ell)$ can all be upper bounded by $\polylog n$,
which implies our algorithm has approximation ratio $\sqrt{d} \cdot \polylog n$ for all these norms.
This clearly demonstrates the generality of our algorithm. 

\paragraph{Empirical Evaluation.} In Appendix~\ref{sec:exp}, we test our algorithms on real datasets. Our empirical results quite clearly demonstrate the practicality of our methods.

\begin{table}
	\caption{Comparison among input-sparsity time linear regression algorithms}
	\center
	\label{table: comparisons of subspace embedding}
	\begin{tabular}{c|c|c}
		\hline
		Reference & Loss Function & Approximation Ratio\\
		\hline
		\hline
		\cite{dasgupta2009sampling, meng2013low, woodruff2013subspace, cohen2015p, wang2019tight} & $\ell_p$ norms & $1+\varepsilon$\\
		\hline
		\cite{yang2013quantile} & Quantile loss function & $1 + \varepsilon$\\
		\hline
		\cite{clarkson2015sketching, cw15b} & $M$-estimators & $1 + \varepsilon$\\
		\hline
		\cite{alszz18} & Orlicz norms& $d \cdot \polylog n$\\
		\hline
		{\bf Theorem \ref{thm:orlicz}}& Orlicz norms& $1+\varepsilon$\\
		\hline
		{\bf Theorem \ref{thm:sym}} & Symmetric norms & $\sqrt{d} \cdot \polylog n \cdot \mmc(\ell)$\\
		\hline
	\end{tabular}
	\vspace{-3mm}
\end{table}

\subsection{Technical Overview}\label{sec:tech}
Similar to previous works on using linear sketching to speed up solving linear regression, our core technique is to provide efficient dimensionality reduction methods for Orlicz norms and general symmetric norms.
In this section, we discuss the techniques behind our results. 
\paragraph{Row Sampling Algorithm for Orlicz Norms.}
Compared to prior work on Orlicz norm regression~\cite{alszz18} which is based on random projection\footnote{Even for $\ell_p$ norms with $p < 2$, embeddings based on random projections will necessarily induce a distortion factor polynomial in $d$, as shown in \cite{wang2019tight}.}, our new algorithm is based on row sampling.
For a given matrix $A\in \mathbb{R}^{n\times d}$,
our goal is to output a {\em sparse} weight vector $w\in \mathbb{R}^{n}$ with at most $\poly(d \log n / \varepsilon)$ non-zero entries, such that with high probability, 
for all $x \in \mathbb{R}^d$, 
\begin{equation}\label{eqn:subspace embedding}
(1 - \varepsilon) \|Ax - b\|_G \le \|Ax - b\|_{G, w} \le (1 + \varepsilon)\|Ax - b\|_{G}.
\end{equation}
Here, for a weight vector $w \in \mathbb{R}^n$ and a vector $y \in \mathbb{R}^n$, the {\em weighted Orlicz norm} $\|y\|_{G, w}$ is defined as the unique value $\alpha$ such that
$\sum_{i=1}^n w_i G(|y_i|/\alpha) = 1$.
See Definition~\ref{def:weighted_orlicz} for the formal definition of weighted Orlicz norm.
To obtain a $(1 + \varepsilon)$-approximate solution to Orlicz norm regression, by \eqref{eqn:subspace embedding}, it suffices to solve 
\begin{equation}\label{equ:weighted_orlicz}
\min_{x \in \R^d} \|Ax - b\|_{G, w}.
\end{equation}

Since the vector $w \in \mathbb{R}^n$ has at most $\poly(d \log n / \varepsilon)$ non-zero entries, and we can ignore all rows of $A$ with zero weights, there are at most $\poly(d \log n / \varepsilon)$ remaining rows in $A$ in the optimization problem in \eqref{equ:weighted_orlicz}.
Furthermore, as we show in Lemma~\ref{lem:is_norm}, $\|\cdot\|_{G, w}$ is a seminorm, which implies we can solve the optimization problem in \eqref{equ:weighted_orlicz} in $\poly(d \log n / \varepsilon)$ time, by simply solving a convex program with size $\poly(d \log n / \varepsilon)$.
Thus, we focus on how to obtain the weight vector $w\in \mathbb{R}^{n}$ in the remaining part.
Furthermore, by taking $\Ab $ to be a matrix whose first $d$ columns are $A$ and last column is $b$, to satisfy \eqref{eqn:subspace embedding}, it suffices to find a weight vector $w$ such that for all $x \in \mathbb{R}^{d + 1}$, 
\begin{equation}\label{eqn:subspace rough embedding}
(1 - \varepsilon) \|\Ab x\|_G \le \|\Ab x\|_{G, w} \le (1 + \varepsilon)\|\Ab x\|_{G}.
\end{equation}
Hence, we ignore the response vector $b$ in the remaining part of the discussion.

We obtain the weight vector $w$ via importance sampling.
We compute a set of sampling probabilities $\{p_i\}_{i=1}^n$ for each row of the data matrix $A$, and sample the rows of $A$ according to these probabilities. 
The $i$-th entry of the weight vector $w$ is then set to be $w_i = 1 / p_i$ with probability $p_i$ and $w_i = 0$ with probability $1-p_i$. 
However, unlike $\ell_p$ norms, Orlicz norms are not ``entry-wise'' norms, and it is not even clear that such a sampling process gives an unbiased estimation.
Our key insight here is that for a vector $Ax$ with unit Orlicz norm, if for all $x \in \mathbb{R}^d$,
\begin{align}
\label{eqn:approx}
(1 - \varepsilon) \sum_{i=1}^n G((Ax)_i) \le \sum_{i=1}^n w_i G((Ax)_i) \le (1 + \varepsilon) \sum_{i=1}^n G((Ax)_i),
\end{align}
then \eqref{eqn:subspace rough embedding} holds, which follows from the convexity of the function $G$.
See Lemma~\ref{lem:sample} and its proof for more details.
Therefore, it remains to develop a way to define and calculate $\{p_i\}_{i=1}^n$, such that the total number of sampled rows is small. 

Our method for defining and computing sampling probabilities $p_i$ is inspired by row sampling algorithms for $\ell_p$ norms \cite{dasgupta2009sampling}.
Here, the key is to obtain an upper bound on the contribution of each entry to the summation $\sum_{i=1}^n G((Ax)_i)$.
Indeed, suppose for some vector $u \in \mathbb{R}^n$ such that $G(Ax)_i \le u_i$ for all $x \in \mathbb{R}^d$ with $\|Ax\|_{G} = 1$, we can then sample each row of $A$ with sampling probability proportional to $u_i$.
Now, by standard concentration inequalities and a net argument, 
\eqref{eqn:approx} holds with high probability.
It remains to upper bound the total number of sampled rows, which is proportional to $\sum_{i = 1}^n u_i$.

We use the case of $\ell_2$ norm, i.e., $G(x) = x^2$, as an example to illustrate our main ideas for choosing the vector $u \in \mathbb{R}^n$.
Suppose $U\in \mathbb{R}^{n\times d}$ is an orthonormal basis matrix of the column space of $A$, then the \emph{leverage score}\footnote{See, e.g., \cite{mahoney2011randomized}, for a survey on leverage scores.} is defined to be the squared $\ell_2$ norm of each row of $U$.
Indeed, leverage score gives an upper bound on the contribution of each row to $\|Ux\|_2^2$, since by Cauchy-Schwarz inequality, for each row $U_{i}$ of $U$, we have
$
\langle U_{i}, x \rangle^2 \le \|U_{i}\|_2^2 \|x\|_2^2 = \|U_{i}\|_2^2 \|Ux\|_2^2
$,
and thus we can set $u_i = \|U_{i}\|_2^2$.
It is also clear that $\sum_{i = 1}^n u_i = d$.

For general Orlicz norms, leverage scores are no longer upper bounds on $G((Ux)_i)$. 
Inspired by the role of orthonormal bases in the case of $\ell_2$ norm, we first define well-conditioned basis for general Orlicz norms as follow. 
\begin{definition}
	\label{def: well-conditioned bases}
	Let $\|\cdot\|_{G}$ be an Orlicz norm induced by a function $G$ which satisfies Assumption~\ref{assump:property_P}. 
	We say $U\in \mathbb{R}^{n\times d}$ is a well-conditioned basis with condition number $\kappa_G = \kappa_G(U)$ if for all $x\in \mathbb{R}^d$, 
	$
	\|x\|_2 \le \|Ux\|_G \le \kappa_G \|x\|_2
	$.
\end{definition}
Given this definition, when $\|Ux\|_G = 1$, by Cauchy-Schwarz inequality and monotonicity of $G$, we can show that 
$
G((Ux)_i) \le G(\|U_{i}\|_2 \|x\|_{2}) \le  G(\|U_{i}\|_2 \|Ux\|_G) \le  G(\|U_{i}\|_2)
$.
This also leads to our definition of Orlicz norm leverage scores.
\begin{definition}\label{def:ls}
	Let $\|\cdot\|_{G}$ be an Orlicz norm induced by a function $G$ which satisfies Assumption~\ref{assump:property_P}. 
	For a given matrix $A \in \mathbb{R}^{n \times d}$ and a well-conditioned basis $U$ of the column space of $A$, the {\em Orlicz norm leverage score} of the $i$-th row of $A$ is defined to be $G(\|U_{i}\|_2)$.
\end{definition}

It remains to give an upper bound on the summation of Orlicz norm leverage scores of all rows.
Unlike the $\ell_2$ norm, it is not immediately clear how to use the definition of well-conditioned basis to obtain such an upper bound for general Orlicz norms.
To achieve this goal, we use a novel probabilistic argument. 
Suppose one takes $x$ to be a vector with i.i.d. Gaussian random variables.
Then each entry of $Ux$ has the same distribution as $\|U_i\|_2 \cdot g_i$, where $\{g_i\}_{i=1}^n$ is a set of standard Gaussian random variables. 
Thus, with constant probability, $\sum_{i=1}^n G((Ux)_i)$ is an upper bound on the summation of Orlicz norm leverage scores.
Furthermore, by the growth condition of the function $G$, we have
$\sum_{i=1}^n G((Ux)_i) \le C_G \|Ux\|_G^2$.
Now by Definition~\ref{def: well-conditioned bases}, $\|Ux\|_G \le \kappa_G \|x\|_2$, and $\|x\|_2 \le O(\sqrt{d})$ with constant probability by tail inequalities of Gaussian random variables. 
This implies an upper bound on the summation of Orlicz norm leverage scores. 
See Lemma~\ref{lem:sum_leverage} and its proof for more details.

Our approach for constructing well-conditioned bases is inspired by~\cite{sohler2011subspace}.
In Lemma \ref{lem:condition}, we show that given a subspace embedding $\Pi$ which embeds the column space of $A$ with Orlicz norm $\|\cdot\|_G$ into the $\ell_2$ space with distortion $\kappa$, then one can construct a well-conditioned basis with condition number $\kappa_G \le \kappa$.
The running time is dominated by calculating $\Pi A$ and doing a QR-decomposition on $\Pi A$.
To this end, we can use the oblivious subspace embedding for Orlicz norms in Corollary~\ref{coro:orlicz_ose}\footnote{Alternatively, we can use the oblivious subspace embedding in \cite{alszz18} for this step. However, as we have discussed, the oblivious subspace embedding in \cite{alszz18} requires stronger assumptions on the function $G : \mathbb{R} \to \mathbb{R}_{\ge 0}$ than those in Assumption~\ref{assump:property_P}, which restricts the class of Orlicz norms to which our algorithm can be applied.} to construct well-conditioned bases.
The embedding in Corollary~\ref{coro:orlicz_ose} has $O(d)$ rows and $\kappa = \poly(d \log n)$, and calculating $\Pi A$ can be done in $\widetilde{O}(\nnz(A) + \poly(d))$ time.
Using such an embedding to construct the well-conditioned basis, our row sampling algorithm produces a vector $w$ that satisfies \eqref{eqn:approx} with $\|w\|_0 \le \poly(d \log n / \varepsilon)$ in time $\wt{O}(\nnz(A) + \poly(d))$.

We would like to remark that our sampling algorithm still works if the third condition in Assumption~\ref{assump:property_P} does not hold.
In general, suppose the function $G : \mathbb{R} \to \mathbb{R}$ satisfies that for all $0<x<y$, $G(y)/G(x)\leq C_G(y/x)^p$, for the Orlicz norm induced by $G$, given a well-conditioned basis with condition number $\kappa_G$, our sampling algorithm returns a matrix with roughly $O((\sqrt{d} \kappa_G)^{p} \cdot d / \varepsilon^2)$ rows such that Theorem~\ref{thm:orlicz} holds. 
One may use the L\"owner–John ellipsoid as the well-conditioned basis (as in~\cite{dasgupta2009sampling}) which has condition number $\kappa_G = \sqrt{d}$ for any norm.
However, calculating the L\"owner–John ellipsoid requires at least $O(nd^5)$ time.
Moreover, our method described above fails when $p > 2$ since it requires an oblivious subspace embedding with $\mathrm{poly}(d)$ distortion, and it is known that such embedding does not exist when $p > 2$~\cite{braverman2010zero}.
Since we focus on input-sparsity time algorithms in this paper, we only consider the case that $p \le 2$.

Finally, we would like to compare our sampling algorithm with that in~\cite{cw15b}.
First, the algorithm in~\cite{cw15b} works for $M$-estimators, while we focus on Orlicz norms.
Second, our definitions for Orlicz norm leverage score and well-conditioned basis, as given in Definition~\ref{def: well-conditioned bases} and \ref{def:ls}, are different from all previous works and are closely related to the Orlicz norm under consideration.
The algorithm in~\cite{cw15b}, on the other hand, simply uses $\ell_p$ leverage scores.  
Under our definition, we can prove that the sum of leverage scores is bounded by $O(C_G d \kappa_G^2)$ (Lemma~\ref{lem:sum_leverage}), whose proof requires a novel probabilistic argument. 
In contrast, the upper bound on sum of leverage scores in~\cite{cw15b} is $O(\sqrt{nd})$ (Lemma 38 in [11]).
Thus, the algorithm in~\cite{cw15b} runs in an iterative manner since in each round the algorithm can merely reduce the dimension from $n$ to $O(\sqrt{nd})$, while our algorithm is one-shot.

\paragraph{Oblivious Subspace Embeddings for Symmetric Norms.}
To obtain a faster algorithm for linear regression when the loss function is a general symmetric norm, we show that there exists a distribution over embedding matrices, such that if $S$ is a random matrix drawn from that distribution, then for any $n\times d$ matrix $A$, with constant probability, for all $x \in \mathbb{R}^d$, 
$
\|Ax\|_{\ell}\leq \|SAx\|_2\leq \poly(d \log n) \cdot \mmc(\ell) \cdot \|Ax\|_{\ell}
$.
Moreover, the embedding matrix $S$ is {\em sparse}, and calculating $SA$ requires only $\wt{O}(\nnz(A) + \poly(d))$ time. 
Another favorable property of $S$ is that it is an {\em oblivious subspace embeeding}, meaning the distribution of $S$ does not depend on $A$. 
To achieve this goal, 
it is sufficient to construct a random diagonal matrix $D$ such that for any fixed vector $x\in\mathbb{R}^n$,
\begin{equation}\label{equ:lb}
\Pr[ \|Dx\|_2\geq \Omega(1 / \poly(d \log n)) \cdot \|x\|_{\ell}] \geq  ~ 1 - \exp(-\Omega(d \log n)), 
\end{equation}
and
\begin{equation}\label{equ:ub}
\Pr[\|Dx\|_2\leq \poly(d \log n) \cdot \mmc(\ell) \cdot \|x\|_{\ell}] \geq ~ 1 - O(1/d).
\end{equation}

Our construction is inspired by the sub-sampling technique in~\cite{iw05}, which was used for sketching symmetric norms in data streams~\cite{blasiok2017streaming}. 
Throughout the discussion, we use $\xi^{(q)}\in\mathbb{R}^n$ to denote a vector with $q$ non-zero entries and each entry is $1/\sqrt{q}$.
Let us start with a special case where the vector $x\in\mathbb{R}^n$ has $s$ non-zero entries and each non-zero entry is $1$. 
It is easy to see $\|x\|_{\ell}= \sqrt{s}\|\xi^{(s)}\|_{\ell}$.
Now consider a random diagonal matrix $D$ which corresponds to a sampling process, i.e., each diagonal entry is set to be $1$ with probability $p$ and $0$ with probability $1 - p$. 
Our goal is to show that $\sqrt{1/p}\|\xi^{(1/p)}\|_{\ell}\cdot\|Dx\|_2$ is a good estimator of $\|x\|_{\ell}$.
If $p=\Theta(d\log n /s)$, then with probability at least $1-\exp\left(-\Omega(d\log n)\right)$, $Dx$ will contain at least one non-zero entry from $x$, in which case \eqref{equ:lb} is satisfied. 
However, we do not know $s$ in advance.
Thus, we use $t=O(\log n)$ different matrices $D_1,D_2,\ldots,D_t$, where $D_i$ has sampling probability $1/2^i$. 
Clearly at least one such $D_j$ can establish \eqref{equ:lb}.
For the upper bound part, if $p$ is much smaller than $1/s$, then $Dx$ will never contain a non-zero entry from $x$.
Otherwise, in expectation $Dx$ will contain $ps$ non-zero entries, in which case our estimation will be roughly  $\sqrt{s}\|\xi^{(1/p)}\|_{\ell}$, which can be upper bounded by $O(\log n \cdot \mmc(\ell) \cdot \sqrt{s}\|\xi^{(s)}\|_{\ell})$.
At this point, \eqref{equ:ub} follows from Markov's inequality.
See Section~\ref{sec:alg_no_dilation_for_each} for the formal argument, and Section~\ref{sec:alg} for a detailed discussion on $\mmc(\ell)$.

To generalize the above argument to general vectors,
for a vector $x\in\mathbb{R}^n$, we conceptually partition its entries into $\Theta(\log n)$ groups, where the $i$-th group contains entries with magnitude in $[2^i,2^{i+1})$. 
By averaging, at least one group of entries contributes at least $\Omega(1/\log n)$ fraction to the value of $\|x\|_{\ell}$. 
To establish \eqref{equ:lb}, we apply the lower bound part of the argument in the previous paragraph to this ``contributing'' group. 
To establish \eqref{equ:ub}, we apply the upper bound part of the argument to all groups, which will only induce an additional $O(\log n)$ factor in the approximation ratio, by triangle inequality.

Since our oblivious subspace embedding embeds a given symmetric norm into the $\ell_2$ space, in order to obtain an approximate solution to symmetric norm regression, we only need to solve a least squares regression instance with much smaller size. This is another advantage of our subspace embedding, since the least square regression problem is a well-studied problem in optimization and numerical linear algebra, for which many efficient algorithms are known, both in theory and in practice.

%% file: orlicz_nips.tex
\section{Linear Regression for Orlicz Norms}
\label{sec:orlicz}
In this section, we introduce our results for Orlicz norm regression. We first give the definition of weighted Orlicz norm.
\begin{definition}\label{def:weighted_orlicz}
	For a function $G$ that satisfies Assumption~\ref{assump:property_P} and a weight vector $w \in \mathbb{R}^n$ such that $w_i \ge  0$ for all $i \in [n]$, 
	for a vector $x \in \mathbb{R}^n$, if $\sum_{i=1}^n w_i \cdot |x_i| = 0$, then the weighted Orlicz norm $\|x\|_{G, w}$ is defined to be $0$. Otherwise, the weighted Orlicz norm $\|x\|_{G, w}$ is defined as the unique value $\alpha > 0$ such that 
	$
	\sum_{i=1}^n w_iG(|x_i|/\alpha) = 1
	$.
\end{definition}

When $w_i = 1$ for all $i \in [n]$, we have $\|x\|_{G, w} = \|x\|_G$ where $\|x\|_{G}$ is the (unweighted) Orlicz norm.
It is well known that $\|\cdot\|_G$ is a norm.
We show in the following lemma that $\|\cdot\|_{G, w}$ is a seminorm.
\begin{lemma}\label{lem:is_norm}
	For a function $G$ that satisfies Assumption~\ref{assump:property_P} and a weight vector $w \in \mathbb{R}^n$ such that $w_i \ge  0$ for all $i \in [n]$,
	for all $x, y \in \mathbb{R}^n$, we have
	(i) $\|x\|_{G, w} \ge 0$, (ii) $\|x + y\|_{G, w} \le \|x\|_{G, w} + \|y\|_{G, w}$, and (iii) $\|ax\|_{G, w} = |a| \cdot \|x\|_{G, w}$ for all $a \in \mathbb{R}$.
\end{lemma}

\paragraph{Leverage Scores and Well-Conditioned Bases for Orlicz Norms.}The following lemma establishes an upper bound on the summation of Orlicz norm leverage scores defined in Definition~\ref{def:ls}.
\begin{lemma}\label{lem:sum_leverage}
Let $\|\cdot\|_{G}$ be an Orlicz norm induced by a function $G$ which satisfies Assumption~\ref{assump:property_P}. 
Let $U\in\mathbb{R}^{n\times d}$ be a well-conditioned basis with condition number $\kappa_G$ as in Definition \ref{def: well-conditioned bases}.
Then we have $\sum_{i=1}^n G(\|U_{i}\|_2) \le O( C_G d  \kappa_G^2)$,
\end{lemma}

Now we show that given a subspace embedding which embeds the column space of $A$ with Orlicz norm $\|\cdot\|_G$ into the $\ell_2$ space with distortion $\kappa$, then one can construct a well-conditioned basis with condition number $\kappa_G \le \kappa$.
\begin{lemma}\label{lem:condition}
Let $\|\cdot\|_{G}$ be an Orlicz norm induced by a function $G$ which satisfies Assumption~\ref{assump:property_P}.
For a given matrix $A \in \mathbb{R}^{n \times d}$ and an embedding matrix $\Pi \in\mathbb{R}^{s\times n}$, suppose for all $x \in \mathbb{R}^d$,
$
\|Ax\|_G \le \|\Pi A x\|_2 \le \kappa \|Ax\|_G
$.
Let $Q \cdot R = \frac{1}{\kappa} \Pi A$ be a QR-decomposition of $\frac{1}{\kappa}  \Pi A$.
Then $AR^{-1}$ is a well-conditioned basis (see Definition \ref{def: well-conditioned bases}) with $\kappa_G(AR^{-1}) \le   \kappa$.
\end{lemma}

The following lemma shows how to estimate Orlicz norm leverage scores given a change of basis matrix $R \in \mathbb{R}^{d \times d}$, in $\widetilde{O}(\nnz(A) + \poly(d))$ time. 

\begin{lemma}\label{lem:score}
Let $\|\cdot\|_{G}$ be an Orlicz norm induced by a function $G$ which satisfies Assumption~\ref{assump:property_P}.
For a given matrix $A \in \R^{n \times d}$ and $R \in \mathbb{R}^{d \times d}$, there exists an algorithm that outputs $\{u_i\}_{i = 1}^n$ such that with probability at least $0.99$, $u_i = \Theta(G(\|(AR^{-1})_{i}\|_2))$ for all $1 \le i \le n$. The algorithm runs in $\widetilde{O}(\nnz(A)+ \poly(d))$ time. 
\end{lemma}

\paragraph{The Row Sampling Algorithm.} Based on the notion of Orlicz norm leverage scores and well-conditioned bases, we design a row sampling algorithm for Orlicz norms.
\begin{lemma}\label{lem:sample}
Let $\|\cdot\|_{G}$ be an Orlicz norm induced by a function $G$ which satisfies Assumption~\ref{assump:property_P}.
Let $U \in \mathbb{R}^{n\times d}$ be a well-conditioned basis with condition number $\kappa_G = \kappa_G(U)$ as in Definition \ref{def: well-conditioned bases}.
For sufficiently small $\varepsilon$ and $\delta$, and sufficiently large constant $C$, 
let $\{p_i\}_{i=1}^n$ be a set of sampling probabilities satisfying
$
 p_i\geq \min\left\{1, C\left(\log(1/\delta) + d \log (1 / \varepsilon)\right)\varepsilon^{-2}  G\left(\|U_{i}\|_{2}\right)
 \right\}$.
Let $w$ be a vector whose $i$-th entry is set to be $w_i=1/p_i$ with probability $p_i$ and $w_i = 0$ with probability $1-p_i$, then with probability at least $1-\delta$, for all $x\in\mathbb{R}^d,$ we have
 $
 (1-\varepsilon)\|Ux\|_{G}\leq \|Ux\|_{G, w}\leq (1+\varepsilon) \|Ux\|_{G}
 $.
\end{lemma}

\paragraph{Solving Linear Regression for Orlicz Norms.} Now we combine all ingredients to give an algorithm for Orlicz norm regression.
We use $\Ab  \in \mathbb{R}^{n \times (d + 1)}$ to denote a matrix whose first $d$ columns are $A$ and the last column is $b$.
The algorithm is described in Figure~\ref{fig:alg_orlicz}, and we prove its running time and correctness in Theorem~\ref{thm:reg_with_se}.
We assume we are given an embedding matrix $\Pi$, such that for all $x \in \mathbb{R}^{d + 1}$,
$
\|\Ab x\|_G \le \|\Pi \Ab  x\|_2 \le \kappa \| \Ab x\|_G
$.
The construction of $\Pi$ and the value $\kappa$ will be given in Corollary~\ref{coro:orlicz_ose}.
In Appendix~\ref{sec:proof_thm1}, we use Theorem~\ref{thm:reg_with_se} and Corollary~\ref{coro:orlicz_ose} to formally prove Theorem~\ref{thm:orlicz}.

\begin{figure}[H]
\begin{framed}
\begin{enumerate}
\item \label{step:QR} For the given embedding matrix $\Pi$, calculate $\Pi \Ab $ and invoke QR-decomposition on $ \Pi \Ab / \kappa $ to obtain $Q \cdot R =  \Pi \Ab / \kappa$.
\item \label{step:obtainp}Invoke Lemma~\ref{lem:score} to obtain $\{u_i\}_{i = 1}^n$ such that $u_i = \Theta(G(\|(\Ab R^{-1})_{i}\|_2))$.
\item \label{step:sample} For a sufficiently large constant $C$, let $\{p_i\}_{i=1}^n$ be a set of sampling probabilities with 
$p_i\geq \min\left\{1, C \cdot  d \cdot \varepsilon^{-2} \log (1 / \varepsilon)   \cdot G\left(\|(\Ab R^{-1})_{i}\|_{2}\right)
 \right\}$ , and $w$ be a vector whose $i$-th entry $w_i=1/p_i$ with probability $p_i$ and $w_i = 0$ with probability $1-p_i$.
 \item \label{step:solvesmall}Calculate $x^* = \argmin_{x \in \mathbb{R}^d} \|Ax - b\|_{G, w}$. Return $x^*$.
\end{enumerate}
\end{framed}
\caption{Algorithm for Orlicz norm regression} \label{fig:alg_orlicz}
\end{figure}

\begin{theorem}\label{thm:reg_with_se}
Let $\|\cdot\|_{G}$ be an Orlicz norm induced by a function $G$ which satisfies Assumption~\ref{assump:property_P}.
Given an embedding matrix $\Pi$, such that for all $x \in \mathbb{R}^d$,
$
\|\Ab x\|_G \le \|\Pi \Ab  x\|_2 \le \kappa \| \Ab x\|_G
$,
with probability at least $0.9$, the algorithm in Figure~\ref{fig:alg_orlicz} outputs $x^* \in \mathbb{R}^d$ in time $\poly(d \kappa / \varepsilon) + \mathcal{T}_{\mathrm{QR}}(\Pi \Ab)$, such that
$
\|Ax^*-b\|_{G} \le (1 + \varepsilon) \min_{x \in \mathbb{R}^d} \|Ax - b\|_{G}
$.
Here, $\mathcal{T}_{\mathrm{QR}}(\Pi \Ab)$ is the running time for calculating $\Pi \Ab$ and invoking QR-decomposition on $\Pi \Ab$.
\end{theorem}

%% file: symmetric_nips.tex
\section{Linear Regression for Symmetric Norms}
\label{sec:alg}
\newcommand{\syms}{\textsf{SymSketch}}
\newcommand{\cous}{\textsf{CountSketch}}
In this section, we introduce \syms, a subspace embedding for symmetric norms. 

\paragraph{Definition of \syms.} We first formally define \syms.
Due to space limitation, we give the definition of Gaussian embeddings, \cous~embeddings and their compositions in Appendix~\ref{sec:l2ose}.
\begin{definition}[Symmetric Norm Sketch (\syms)]\label{def:symmetric_sketch}
Let $t=\Theta(\log n)$. Let $\wt{D} \in \R^{n(t+1) \times n}$ be a matrix defined as
	$
	\wt{D}= \begin{bmatrix}
	(w_0D_0)^\top &
	(w_1D_1)^\top &
	\ldots &
	(w_tD_t)^\top
	\end{bmatrix}^\top
	$,
	where for each $i \in \{0,1,\ldots,t\}$, $D_i = \diag(z_{i,1}, z_{i,2}, \ldots, z_{i,n}) \in \R^{n\times n}$ and $z_{i,j} \sim \mathrm{Ber}(1/2^i)$ for each $j\in  [n]$. 
	Moreover, 
	$
	w_i = \| ( 1,1,\ldots, 1 ,0,\ldots, 0 ) \|_{\ell}
	$ (there are $2^i$ $1$s).
	Let $\Pi \in \R^{O(d) \times  n(t+1)}$ be a composition of Gaussian embedding and \cous~embedding (Definition~\ref{def:fast_gaussian_transform}) with $\varepsilon = 0.1$, and $S=\Pi \wt{D}$. 
	We say $S\in \R^{O(d) \times n}$ is a $\syms$.
\end{definition}

\paragraph{Modulus of Concentration.} Now we give the definition of $\mmc(\ell)$ for a symmetric norm.

\begin{definition}[\cite{blasiok2017streaming}]\label{def:median}
Let ${\cal X}$ denote the uniform distribution over $\mathbb{S}^{n - 1}$.
The {\em median} of a symmetric norm $\|\cdot\|_{\ell}$ is the unique value $M_{\ell}$ such that 
$\Pr_{x\sim {\cal X} }[ \|x\|_{\ell} \geq M_{\ell} ] \geq 1/2$ 
and 
$\Pr_{x \sim {\cal X} }[ \|x\|_{\ell} \leq M_{\ell} ] \geq 1/2$.
\end{definition}

\begin{definition}[\cite{blasiok2017streaming}]\label{def:mmc}
	For a given symmetric norm $\|\cdot\|_{\ell}$,
        we define the \emph{modulus of concentration} to be 
        $
        \mc(\ell) = \max_{x \in  \mathbb{S}^{n-1}} \|x\|_{\ell} / M_{\ell}
        $,
	and define the {\em maximum modulus of concentration} to be
	$
	\mmc(\ell) = \max_{k \in [n] } \mc( \ell^{(k)})
	$,
	where $\|\cdot\|_{{\ell}^{(k)}}$ is a norm on $\R^k$ which is defined to be
	$
	\| (x_1,x_2,\ldots,x_k) \|_{{\ell}^{(k)}} =  \| (x_1,x_2,\ldots,x_k,0,\ldots,0) \|_{\ell}
	$.
\end{definition}

It has been shown in \cite{blasiok2017streaming} that $\mmc(\ell) = \Theta(n^{1/2-1/p})$ for $\ell_p$ norms when $p>2$, $\mmc(\ell)=\Theta(1)$ for $\ell_p$ norms when $p \le 2$, $\mmc(\ell) = \wt{\Theta}(\sqrt{n/k})$ for top-$k$ norms, and $\mmc(\ell) = O(\log n)$ for  the $k$-support norm \cite{argyriou2012sparse} and the box-norm \cite{mcdonald2014spectral}.
We show that $\mmc(\ell)$ is upper bounded by $O(1)$ for max-mix of $\ell_2$ norm and $\ell_1$ norm and sum-mix of $\ell_2$ norm and $\ell_1$ norm.\begin{lemma}\label{lem:example_mmc}
For a real number $c > 0$,
let $\|x\|_{\ell_a}=\|x\|_2+ c\|x\|_1$
and $\|x\|_{\ell_b}=\max\{\|x\|_2, c\|x\|_1\}$.
We have $\mmc(\ell_a) = O(1)$ and $\mmc(\ell_b)= O(1)$.
\end{lemma}

Moreover, we show that for an Orlicz norm $\|\cdot\|_{G}$ induced by a function $G$ which satisfies Assumption~\ref{assump:property_P}, $\mmc(\ell)$ is upper bounded by $ O(\sqrt{C_G} \log n)$.

\begin{lemma}\label{lem:orlicz_mmc}
	For an Orlicz norm $\|\cdot\|_{G}$ on $\R^n$ induced by a function $G$ which satisfies Assumption~\ref{assump:property_P}, $\mmc(\ell)$ is upper bounded by $ O(\sqrt{C_G} \log n)$.
\end{lemma}

\paragraph{Subspace Embedding.} The following theorem shows that \syms~is a subspace embedding.
\begin{theorem}\label{thm:subspace_embeddings_formal}
	Let $S\in\mathbb{R}^{O(d)\times n}$ be a \syms~as defined in Definition~\ref{def:symmetric_sketch}.
	For a given matrix $A\in \R^{n\times d}$, with probability at least $0.9$, for all $x \in \R^d$, 
	\begin{align*}
	\Omega \left( 1/ ( \sqrt{d} \cdot \log^3 n ) \right)\cdot \| A x \|_{\ell} \leq \| S A x \|_2 \leq O \left( \mmc(\ell) \cdot d^2 \cdot \log^{5/2} n\right)\cdot \| A x \|_{\ell}.
	\end{align*}
	Furthermore, the running time of computing $SA$ is $\wt{O}(\nnz(A)+\poly(d))$.
\end{theorem}

Combine Theorem~\ref{thm:subspace_embeddings_formal} with Lemma~\ref{lem:orlicz_mmc}, we have the following corollary.
\begin{corollary}\label{coro:orlicz_ose}
Let $\|\cdot\|_{G}$ be an Orlicz norm induced by a function $G$ which satisfies Assumption~\ref{assump:property_P}.
Let $S\in\mathbb{R}^{O(d)\times n}$ be a \syms~as defined in Definition~\ref{def:symmetric_sketch}.
For a given matrix $A\in \R^{n\times d}$, with probability at least $0.9$, for all $x \in \R^d$, 
	\begin{align*}
	\Omega\left( 1 /( \sqrt{d} \cdot \log^3 n )\right)\cdot \| A x \|_{\ell} \leq \| S A x \|_2 \leq O\left( \sqrt{C_G} \cdot d^2 \cdot \log^{7/2} n\right)\cdot \| A x \|_{\ell}.
	\end{align*}
	Furthermore, the running time of computing $SA$ is $\wt{O}(\nnz(A)+\poly(d))$.
\end{corollary}

%% file: conclusion.tex
\section{Conclusion}
In this paper, we give efficient algorithms for solving the overconstrained linear regression problem, when the loss function is a symmetric norm.
For the special case when the loss function is an Orlicz norm, our algorithm produces a $(1 + \varepsilon)$-approximate solution in $\wt{O}(\nnz(A) + \poly(d / \varepsilon))$ time.
When the loss function is a general symmetric norm, our algorithm produces a $\sqrt{d} \cdot \polylog n \cdot \mmc(\ell)$-approximate solution in $\wt{O}(\nnz(A) + \poly(d))$ time.

In light of Problem~\ref{question:universal}, there are a few interesting problems that remain open.
Is that possible to design an algorithm that produces $(1 + \varepsilon)$-approximate solutions to the linear regression problem, when the loss function is a general symmetric norm?
Furthermore, is that possible to use the technique of linear sketching to speed up the overconstrained linear regression problem, when the loss function is a general norm?
Answering these problems could lead to a better understanding of Problem~\ref{question:universal}.

%% file: ack.tex
\section*{Acknowledgements}
P. Zhong is supported in part by NSF grants (CCF-1703925, CCF-1421161, CCF-1714818, CCF-1617955 and CCF-1740833), Simons Foundation (\#491119), Google Research Award and a Google Ph.D. fellowship.
R. Wang is supported in part by NSF grant IIS-1763562, Office of Naval Research (ONR) grants (N00014-18-1-2562, N00014-18-1-2861), and Nvidia NVAIL award.
Part of this work was done 
while Z. Song, L. F. Yang, H. Zhang and P. Zhong were interns at IBM Research - Almaden and 
while Z. Song, R. Wang and H. Zhang were visiting the Simons Institute for the Theory of Computing.
Z. Song and P. Zhong would like to thank Alexandr Andoni, Kenneth L. Clarkson, Yin Tat Lee, Eric Price, Clifford Stein and David P. Woodruff for insight discussions. 

%% file: pre.tex
\section{Preliminaries}\label{sec:pre}

\paragraph{Notations.}
For a matrix $A \in \mathbb{R}^{n \times d}$, we use $A_{i}$ to denote its
$i$-th row, $A^i$ to denote its $i$-th column, $\|A\|_F$ to denote the Frobenius norm of $A$, and $\|A\|_2$ to denote the spectral norm of $A$.
For any $n'\leq n,$ we define $\xi^{(n')}\in\mathbb{R}^n$ to be a vector $\xi^{(n')}=\frac{1}{\sqrt{n'}}(1,1,\ldots,1,0,0,\ldots,0)$.
\paragraph{$\varepsilon$-nets.}
We use the standard upper bound on size of $\varepsilon$-nets.

\begin{definition}\label{def:net}
For a given set $\mathcal{S}$ and a norm $\| \cdot \|$, we say $\mathcal{N} \subseteq \mathcal{S}$ is a $\varepsilon$-net of $\mathcal{S}$ if for any $s \in \mathcal{S}$, there exists some $\overline{s} \in \mathcal{N}$ such that $\|s - \overline{s}\| \le \varepsilon$.
\end{definition}
\begin{lemma}[{\cite[II.E,~10]{wojtaszczyk1996banach}}]\label{lem:net_size}
	Given a matrix $A\in\mathbb{R}^{n\times d}$ and a norm $\|\cdot\|$,
	let $\mathcal{S}$ be the unit $\|\cdot\|$-norm ball in the column space of $A$, i.e., $\mathcal{S}=\{Ax\mid \|Ax\|=1\}.$
	For $\varepsilon\in (0,1),$ there exists an $\varepsilon$-net $\mathcal{N}$ of $\cal{S}$ with size $|\mathcal{N}|\leq (1+1/\varepsilon)^d.$
\end{lemma}

%% file: orlicz_appendix.tex
\section{Missing Proofs in Section~\ref{sec:orlicz}}
In this section, we give missing proofs in Section~\ref{sec:orlicz}.

We first show that if a function $G$ satisfies Assumption~\ref{assump:property_P}, then $G$ has at least linear growth. 
We will use this fact in later proofs.
\begin{lemma}\label{lem:fasterthanlinear}
Given a function $G$ that satisfies property $\mathcal{P}$, then for any $0 < x \le y$, $y / x \le G(y) / G(x)$.
\end{lemma}

\begin{proof}
	Due to the convexity of $G$ and $G(0) = 0$, for any $y > x > 0$, we have
	\[
	G(x) \le  G(y) x / y + G(0) (1 - x / y) = G(y)x/y.
	\]
\end{proof}

\subsection{Proof of Lemma~\ref{lem:is_norm}}
\begin{proof}
	The first condition is clear from the definition of $\|x\|_{G, w}$.
	
	Now we prove the second condition. 
	When $\|x + y\|_{G, w} = 0$, the triangle inequality clearly holds since $\|x\|_{G, w} \ge 0$ and $\|y\|_{G, w} \ge 0$.
	When $\|x\|_{G, w} = 0$ and $\|x + y\|_{G, w} \neq 0$, for any $\alpha > 0$, we have
	\[
	\sum_{i=1}^n w_iG(|x_i + y_i|/\alpha) = \sum_{i \mid w_i > 0} w_iG(|x_i + y_i|/\alpha) = \sum_{i \mid w_i > 0} w_iG(| y_i|/\alpha) = \sum_{i=1}^n w_iG(|y_i|/\alpha),
	\]
	which implies $\|x + y\|_{G, w} = \|y\|_{G, w}$.
	Similarly, the second condition also holds if $\|y\|_{G, w} = 0$ and $\|x + y\|_{G ,w} \neq 0$.
	If $\|x + y\|_{G, w} \neq 0$, $\|x\|_{G, w} \neq 0$ and $\|y\|_{G, w} \neq 0$, by definition of $\|\cdot\|_{G, w}$, we have
	\[
	\sum_{i=1}^n w_i G(x_i / \|x\|_{G, w}) = 1
	\]
	and
	\[
	\sum_{i=1}^n w_i G(y_i / \|y\|_{G, w}) = 1.
	\]
	Thus,
	\begin{align*}
	&\sum_{i=1}^n w_i G\left( \frac{x_i + y_i} {\|x\|_{G, w} + \|y\|_{G, w}} \right) \\
	\le & \sum_{i=1}^n w_i G\left( \frac{|x_i| + |y_i|} {\|x\|_{G, w} + \|y\|_{G, w}} \right) \tag{$G$ is increasing} \\
	\le & \sum_{i=1}^n w_i \left( \frac{ \|x\|_{G, w}}{\|x\|_{G, w} + \|y\|_{G, w}} \cdot G(|x_i| / \|x\|_{G, w})  +  \frac{ \|y\|_{G, w}}{\|x\|_{G, w} + \|y\|_{G, w}} \cdot G(|y_i| / \|y\|_{G, w}) \right) \tag{$G$ is convex} \\
	= & 1,
	\end{align*}
	which implies $\|x + y\|_{G, w} \le \|x\|_{G, w} + \|y\|_{G, w}$.
	
	For the third condition, for any $a \in \mathbb{R}$ and $x \in \mathbb{R}^n$, if $\|x\|_{G, w} = 0$ then $\|ax\|_{G, w} = 0$. If $a = 0$, we have $\|ax\|_{G, w} = 0$. 
	Otherwise, we have
	\[
	\sum_{i=1}^n w_i G(x_i / \|x\|_{G, w}) = 1,
	\]
	which implies
	\[
	\sum_{i=1}^n w_i G\left( \frac{a \cdot x_i}{ |a|\|x\|_{G, w}} \right) = 1,
	\]	
	and thus $\|ax\|_{G, w} = |a| \|x\|_{G, w}$.
\end{proof}

\subsection{Proof of Lemma~\ref{lem:sum_leverage}}

\begin{proof}
	Let $g \in \mathbb{R}^d$ be a vector whose entries are i.i.d. Gaussian random variables with zero mean and standard deviation $10^2$.
	We show that with probability at least $0.8$,
	\[
	\sum_{i=1}^n G(\|U_{i}\|_2) \le O \left( \sum_{i=1}^n G(\langle U_{i}, g \rangle) \right) \le O \left( \max\{1, C_G \|Ug\|_G^2 \} \right) \le O( C_G d \kappa_G^2).
	\] 
	
	We divide our proofs into three parts.
	\paragraph{Part I.} We will show that with probability at least $0.9$, 
	\[
	\sum_{i=1}^n G(\|U_{i}\|_2) \le O \left( \sum_{i=1}^n G(\langle U_{i}, g \rangle) \right).
	\]
	For each $i \in [n]$, $\langle U_{i}, g \rangle$ has the same distribution as $10^2 \cdot \|U_i\|_2 \cdot \mathcal{N}(0, 1)$.
	For each $i \in [n]$, we let $B_i$ be the random variable such that
	\[
	B_i = \begin{cases}
	1 & |\langle U_{i}, g \rangle| \le \|U_i\|_2 \\
	0 & \text{otherwise}
	\end{cases}.
	\]
	By tail inequalities of standard Gaussian random variables, $\Pr[B_i = 1] \le 0.01$.
	Thus, 
	\[
	\E\left[B_i \cdot G(\|U_i\|_2)\right] \le 0.01 \cdot G(\|U_i\|_2),
	\]
	which implies
	\[
	\E\left[\sum_{i = 1}^n B_i \cdot G(\|U_i\|_2)\right] \le 0.01 \cdot \sum_{i= 1}^nG(\|U_i\|_2),
	\]
	
	By the monotonicity of $G$, since
	\[
	G(\langle U_{i}, g) \rangle \ge (1 - B_i) G(\|U_i\|_2),
	\]
	we have
	\[
	\sum_{i=1}^n G(\langle U_{i}, g \rangle) \ge \sum_{i=1}^n (1 - B_i) G(\|U_i\|_2).
	\]
	
	By Markov's inequality, with probability at least $0.9$, we have
	\[
	\sum_{i=1}^nB_i \cdot G(\|U_i\|_2)  \le 0.1 \sum_{i=1}^n G(\|U_i\|_2),
	\]
	which implies
	\[
	\sum_{i=1}^n G(\langle U_{i}, g \rangle) \ge 0.9 \sum_{i=1}^n G(\|U_i\|_2).
	\]
	
	\paragraph{Part II.} We will show that 
	\[
	\sum_{i=1}^n G(\langle U_{i}, g \rangle) \le \max\{1, C_G \cdot \|Ug\|_G^2\}.
	\]
	When $\|Ug\|_G \le 1$, by monotonicity of $G$, we must have
	\[
	\sum_{i=1}^n G(\langle U_{i}, g \rangle)  \le 1.
	\]
	When $\|Ug\|_G \ge 1$, we have
	\[
	\sum_{i=1}^n G(\langle U_{i}, g  \rangle/ \|Ug\|_G) = 1.
	\]
	Since 
	\[
	G(\langle U_{i}, g \rangle) \le G(\langle U_{i}, g  \rangle/ \|Ug\|_G) \cdot C_G \|Ug\|_G^2
	\]
	and
	\[
	\sum_{i=1}^n G(\langle U_{i}, g  \rangle/ \|Ug\|_G)  = 1,
	\]
	we must have
	\[
	\sum_{i = 1}^n G(\langle U_{i}, g \rangle) \le \sum_{i=1}^n G(\langle U_{i}, g  \rangle/ \|Ug\|_G) \cdot C_G \cdot \|Ug\|_G ^2 \le C_G \cdot \|Ug\|_G ^2.
	\]
	
	\paragraph{Part III.} We will show that $\|Ug\|_G^2 \le O(C_G d \kappa_G^2)$.
	By definition of a well-conditioned basis and tail inequalities of Gaussian random variables, with probability at least $0.9$, we have
	\[
	\|Ug\|_G \le \kappa_G \|g\|_2 \le O(\kappa_G \sqrt{d}).
	\]
	
	Thus, applying a union bound over three parts of the proof, we have with probability at least $0.8$,
	\begin{equation} \label{equ:sum_ls}
	\sum_{i=1}^n G(\|U_{i}\|_2) \le O( C_G d \kappa_G^2).
	\end{equation}
	However, the condition in \eqref{equ:sum_ls} is deterministic. 
	Thus, the condition in \eqref{equ:sum_ls} always holds. 
\end{proof}
\subsection{Proof of Lemma~\ref{lem:condition}}
\begin{proof}
Notice that for any $x \in \mathbb{R}^d$,
\[
\|AR^{-1}x\|_G \le \|\Pi AR^{-1}x\|_2 = \kappa \|Qx\|_2 = \kappa \|x\|_2
\]
and
\[
\|AR^{-1}x\|_G  \ge \frac{1}{\kappa} \|\Pi AR^{-1}x\|_2 = \|Qx\|_2 = \|x\|_2.
\]
\end{proof}
\subsection{Proof of Lemma~\ref{lem:score}}
\begin{proof}
In Theorem 2.13 of \cite{woodruff2014sketching}, it has been shown how to calculate $\{l_i\}_{i = 1}^n$ such that $l_i = \Theta(\|(AR^{-1})_{i}\|_2)$ in $\widetilde{O}(\nnz(A)+ \poly(d))$ time with probability at least $0.99$.
We simply take $u_i = G(l_i)$.
By Lemma \ref{lem:fasterthanlinear} and the growth condition of $G$, we must have $u_i = \Theta(G(\|(AR^{-1})_{i}\|_2))$.
\end{proof}

\subsection{Proof of Lemma~\ref{lem:sample}}
\begin{proof}
	By homogeneity, we only need to prove that with probability $1-\delta$, for all $x$ which satisfies $\|Ux\|_{G}=1,$
	\begin{align*}
	(1-\varepsilon)\|Ux\|_{G}\leq \|Ux\|_{G, w}\leq (1+\varepsilon) \|Ux\|_{G}.
	\end{align*}
	
	We first prove that for any fixed $x \in \mathbb{R}^d$ such that $\|Ux\|_{G}=1$, with probability $1 - \delta (1 + 4 / \varepsilon)^{-d}$,
	\begin{align*}
	(1-\varepsilon / 4)\|Ux\|_{G}\leq \|Ux\|_{G, w}\leq (1+\varepsilon / 4) \|Ux\|_{G}.
	\end{align*}
	
	Let $x\in\mathbb{R}^d$ that satisfies $\|Ux\|_{G}=1$ and $y=Ux$.
	Let $Z_i$ be a random variable which denotes the value of $w_i G(y_i)$ and $Z=\sum_{i=1}^n Z_i$.
	
	We will first show that if $Z\in[1-\varepsilon / 4,1+\varepsilon / 4],$ then $\|y\|_{G, w}\in[1-\varepsilon / 4,1+\varepsilon / 4]$. There are three cases:
	\begin{enumerate}
		\item If $\|y\|_{G, w}=1,$ then $\|y\|_{G, w}$ is already in $[1-\varepsilon / 4,1+\varepsilon / 4]$.
		\item If $\|y\|_{G, w}>1,$ then by Lemma~\ref{lem:fasterthanlinear}, we have 
		\[
		\sum_{i=1}^n w_i G(y_i)\geq \sum_{i=1}^n w_i \|y\|_{G, w}\cdot G(y_i/\|y\|_{G, w}).
		\]
		Since 
		\[
		 \sum_{i=1}^n w_i \cdot G(y_i/\|y\|_{G, w}) = 1,
		 \]
		we must have
		\[
		\|y\|_{G, w}\leq \sum_{i=1}^n w_i G(y_i)=Z\leq 1+\varepsilon / 4.
		\]
		\item If $\|y\|_{G, w}<1,$ then by Lemma~\ref{lem:fasterthanlinear}, we have 
		\[
		1 = \sum_{i=1}^n w_i G(y_i/\|y\|_{G, w})\geq 1/\|y\|_{G, w} \cdot \sum_{i=1}^n w_i  G(y_i),\]
		 which implies 
		\[
		\|y\|_{G, w}\geq \sum_{i=1}^n w_i  G(y_i)=Z\geq 1-\varepsilon / 4.
		\]
	\end{enumerate}
	
	Thus, it suffices to prove that \[
	\Pr\left[Z\in[1-\varepsilon / 4,1+\varepsilon / 4] \right]\geq 1-\delta(1+4/\varepsilon)^{-d}.\]
	Consider the expectation of $Z$, we have
	\begin{align*}
	\E[Z]&=\sum_{i=1}^n \E[Z_i] = \sum_{i=1}^n \E[ w_i] \cdot G((Ux)_i)= \sum_{i=1}^n  G((Ux)_i) = 1,
	\end{align*}
	where the last equality follows since $\|Ux\|_{G}=1.$
	
	Notice that $|Z_i-\E(Z_i)|$ is always upper bounded by
	\begin{align*}
	w_i G(y_i) = w_i G((Ux)_i) &\leq w_i G(\|U_{i}\|_{2}\cdot \|x\|_{2}) \leq w_i G\left(\|U_{i}\|_{2}\right)\\
	&\leq  G\left(\|U_{i}\|_{2}\right)/p_i\leq \frac{\varepsilon^2}{C\left(\log(1/\delta) + d \log (1 / \varepsilon)\right)},
	\end{align*}
	where the first inequality follows from Cauchy-Schwarz inequality, the second inequality follows from the definition of well-conditioned basis in Definition~\ref{def: well-conditioned bases} and monotonicity of $G$, the third inequality follows from definition of $w_i$ and the last inequality follows from the choice of $p_i.$
	
	Consider the variance of $Z$, we have:
	\begin{align*}
	\var(Z)&
	=\sum_{i\mid p_i<1} \var(Z_i)
	\leq \sum_{i\mid p_i<1} \E(Z_i^2) 
	=\sum_{i\mid p_i<1} \left(G((Ux)_i)\right)^2 / p_i\\
	& \le \left( \sum_{i \mid p_i < 1} G((Ux)_i) \right) \cdot \max_{i \mid p_i < 1} G((Ux)_i) / p_i \le \frac{\varepsilon^2}{C\left(\log(1/\delta) + d \log (1 / \varepsilon)\right)},
	\end{align*}
	where the second inequality follows from H\"older's inequality and the last inequality follows from the upper bound of $G((Ux)_i) / p_i$ and $\|Ux\|_G = 1$.
	
	Thus, by Bernstein inequality, we have:
	\[
	\Pr\left(|Z-1|>\varepsilon / 4\right) \leq (1 + 4 / \varepsilon)^{-d}\delta.
	\]
	Thus, for a fixed $x,$ with probability at least $1-(1 + 4 / \varepsilon)^{-d}\delta,$ we have
	\[ (1-\varepsilon / 4)\|Ux\|_{G}\leq \|Ux\|_{G, w}\leq (1+\varepsilon / 4) \|Ux\|_{G}.\]
	
	
	
	Let $\mathcal{S}$ be the unit $\|\cdot\|_{G}$-norm ball in the column space of $U$, i.e., $\mathcal{S}=\{Ux\mid \|Ux\|_{G}=1\}.$
	According to Lemma~\ref{lem:net_size}, there exists an $\varepsilon/4$-net $\mathcal{N}$ of $\mathcal{S}$ with $|\mathcal{N}|\leq (1+4/\varepsilon)^d.$
	We use $\mathcal{E}$ to denote the event that for all $y \in \mathcal{N}$, $\|y\|_{G, w}\in [1-\varepsilon/4,1+\varepsilon/4]$.
	By taking union bound over all vectors in $\mathcal{N}$, we have $
	\Pr[\mathcal{E}]\geq 1 - \delta.
	$
	
	Conditioned on $\mathcal{E},$ now we show that for all $y\in \mathcal{S},\|y\|_{G, w}\in[1-\varepsilon,1+\varepsilon]$.
	Consider a fixed vector $y\in \mathcal{S},$ since $\mathcal{N}$ is an $\varepsilon/4$-net of $\mathcal{S},$ we can choose a vector $u^{(1)}\in\mathcal{N}$ such that
	\[
	\|y-u^{(1)}\|_{G}\leq \varepsilon/4.
	\]
	Thus, we have that 
	\[
	\|y\|_{G, w}\leq \|u^{(1)}\|_{G, w}+\|y-u^{(1)}\|_{G, w}\leq (1+\varepsilon/4)+\|y-u^{(1)}\|_{G, w}.
	\]
	Let $\alpha^{(1)}=1/\|y-u^{(1)}\|_{G}.$
	Then we have $\alpha^{(1)}(y-u^{(1)})\in\mathcal{S}.$
	Thus, there exist $u^{(2)}\in\mathcal{N}$ such that 
	\[
	\|u^{(2)}-\alpha^{(1)}(y-u^{(1)})\|_{G}\leq \varepsilon/4.
	\]
	It implies that 
	\[
	\|(y-u^{(1)})-u^{(2)}/\alpha^{(1)}\|_{G}\leq \varepsilon/(4\alpha^{(1)})\leq (\varepsilon/4)^2.
	\]
	Thus, 
	\[
	\|y-u^{(1)}\|_{G, w}\leq \|u^{(2)}\|_{G, w}/\alpha^{(1)}+\|y-u^{(1)}-u^{(2)}/\alpha^{(1)}\|_{G, w}\leq (1+\varepsilon/4)\varepsilon/4+\|y-u^{(1)}-u^{(2)}/\alpha^{(1)}\|_{G, w}.
	\]
	Let $\alpha^{(2)}=1/\|y-u^{(1)}-u^{(2)}/\alpha^{(1)}\|_{G}.$
	Then we can repeat the above argument and get
	\begin{align*}
	\|y\|_{G, w}&\leq (1+\varepsilon/4)+(1+\varepsilon/4)\varepsilon/4+(1+\varepsilon/4)(\varepsilon/4)^2+\ldots\\
	&= (1+\varepsilon/4)/(1-\varepsilon/4)\leq 1+\varepsilon.
	\end{align*}
	By applying the above upper bound on $\|\alpha^{(1)}(u^{(1)}-y)\|_{G,w}$, we can get
	\begin{align*}
	\|y\|_{G, w}&\geq \|u^{(1)}\|_{G, w}-\|u^{(1)}-y\|_{G, w}\\
	&\geq(1-\varepsilon/4)-\|u^{(1)}-y\|_{G, w}\\
	&\geq (1-\varepsilon/4)-\frac{1+\varepsilon}{\alpha^{(1)}}\\
	&\geq 1-\varepsilon/2-\varepsilon^2/4\\
	&\geq 1-\varepsilon.
	\end{align*}
	Thus, conditioned on $\mathcal{E}$, which holds with probability $1 - \delta$, we have $\|y\|_{G, w}\in[1-\varepsilon,1+\varepsilon]$ for all $y = Ux$ with $\|y\|_G = 1$.
\end{proof}

\subsection{Proof of Theorem~\ref{thm:reg_with_se}}
\begin{proof}

We first analyze the running time of the algorithm. 
In Step~\ref{step:QR}, we calculate $\Pi \Ab $ and invoke QR-decomposition on $\Pi \Ab $.
In Step~\ref{step:obtainp}, we apply the algorithm in Lemma~\ref{lem:score}, which runs in $\widetilde{O}(\nnz(A) + \poly(d))$ time.
Obtaining the weight vector $w \in \mathbb{R}^n$ in Step~\ref{step:sample} requires $O(n)$ time.

Since for all $x \in \mathbb{R}^d$,
\[
\|\Ab x\|_G \le \|\Pi \Ab  x\|_2 \le \kappa \| \Ab x\|_G.
\]
we have
\begin{align*}
\E[\|w\|_0]  
& = \sum_{i = 1}^{n} p_i = \sum_{i = 1}^{n} O \left( d \log ( 1 / \varepsilon) / \varepsilon^2  \cdot G\left(\|(\Ab R^{-1})_{i}\|_{2}\right) \right)\\
& \le O\left( d\log ( 1 / \varepsilon) / \varepsilon^2 \cdot C_G d \left(\kappa_G(\Ab R^{-1})\right)^2 \right)  \tag{Lemma~\ref{lem:sum_leverage}}\\
& \le O\left( C_G d^2 \kappa^2 \log ( 1 / \varepsilon) / \varepsilon^2 \right).  \tag{Lemma~\ref{lem:condition}}
\end{align*}
By Markov's inequality, with constant probability we have $\|w\|_0 \le O\left( C_G d^2 \kappa^2 \log ( 1 / \varepsilon) / \varepsilon^2 \right)$.
Moreover, in order to solve $\min_x \|Ax - b\|_{G, w}$, we can ignore all rows of $A$ with zero weights, and thus there are at most $O\left( C_G d^2 \kappa^2 \log ( 1 / \varepsilon) / \varepsilon^2 \right)$ remaining rows in $A$.
Furthermore, as we show in Lemma~\ref{lem:is_norm}, $\|\cdot\|_{G, w}$ is a seminorm, which implies we can solve $\min_x \|Ax - b\|_{G, w}$ in $\poly(C_Gd \kappa / \varepsilon)$ time, by simply solve a convex program with size $O\left( C_G d^2 \kappa^2 \log ( 1 / \varepsilon) / \varepsilon^2 \right)$.

Now we prove the correctness of the algorithm. 
The algorithm in Lemma~\ref{lem:score} succeeds with constant probability. By Lemma~\ref{lem:sample}, with constant probability, simultaneously for all $x \in \mathbb{R}^{d + 1}$,
\[
 (1-\varepsilon / 3)\|\Ab R^{-1}x\|_{G}\leq \|\Ab R^{-1}x\|_{G, w}\leq (1+\varepsilon / 3) \|\Ab R^{-1}x\|_{G}.
\]
Equivalently, with constant probability, simultaneously for all $x \in \mathbb{R}^{d + 1}$,
\[
 (1-\varepsilon / 3)\|\Ab x\|_{G}\leq \|\Ab x\|_{G, w}\leq (1+\varepsilon / 3) \|\Ab x\|_{G}.
\]

Since $x^* = \argmin_x \|Ax - b\|_{G, w}$, for all $x \in \mathbb{R}^d$, we have
\begin{align*}
\|Ax^* - b\|_{G} &\le 1  / (1 - \varepsilon / 3) \|Ax^* - b\|_{G, w} 
\le 1 / (1 - \varepsilon / 3) \|Ax - b\|_{G, w} \\
&\le (1 + \varepsilon / 3) / (1 - \varepsilon / 3) \|Ax - b\|_{G} \le (1 + \varepsilon) \|Ax - b\|_{G}
\end{align*}
for sufficiently small $\varepsilon$.
Thus, $x^*$ is a $(1 + \varepsilon)$-approximate solution to $\min_{x} \|Ax - b\|_G$.

Note that the failure probability of the algorithm can be reduced to an arbitrarily small constant by independent repetitions and taking the best solution found among all repetitions.
\end{proof}

%% file: symmetric_appendix.tex
\section{Missing Proofs in Section~\ref{sec:alg}}

In this section, we give missing proofs in Section~\ref{sec:alg}.

Without loss of generality, throughout this section, for the symmetric norm $\|\cdot\|_{\ell}$ under consideration, we assume $\|\xi^{(1)}\|_{\ell}=1.$

\subsection{Background}

\subsubsection{Known $\ell_2$ Oblivious Subspace Embeddings}\label{sec:l2ose}
In this section, we recall some known $\ell_2$ subspace embeddings.

\begin{definition}
We say $S \in \mathbb{R}^{t \times n}$ is an $\ell_2$ subspace embedding for the column space of $A \in \mathbb{R}^{n \times d}$, if for all $x\in \R^d$,
\begin{align*}
(1 - \epsilon) \| A x\|_2  \leq \| S A x \|_2 \leq (1 + \epsilon) \| A x\|_2.
\end{align*}
\end{definition}

\begin{definition}\label{def:count_sketch_transform}
A \cous~embedding is defined to be $\Pi= \Phi D\in \mathbb{R}^{m\times n}$ with $m=\Theta(d^2/\varepsilon^2)$, where $D$ is an $n\times n$ random diagonal matrix with each diagonal entry independently chosen to be $+1$ or $-1$ with equal probability, and $\Phi\in \{0,1\}^{m\times n}$ is an $m\times n$ binary matrix with $\Phi_{h(i),i}=1$ and all remaining entries being $0$, where $h:[n]\to [m]$ is a random map such that for each $i\in [n]$, $h(i) = j$ with probability $1/m$ for each $j \in [m]$. 
\end{definition}

\begin{theorem}[\cite{clarkson2013low}]
For a given matrix $A \in \mathbb{R}^{n \times d}$ and $\varepsilon\in(0,1/2)$. 
Let $\Pi\in\mathbb{R}^{\Theta(d^2/\varepsilon^2) \times n}$ be a \cous~embedding. With probability at least $0.9999$, $\Pi$ is an $\ell_2$ subspace embedding for the column space of $A$.
Furthermore, $\Pi A$ can be computed in $O(\nnz(A))$ time. 
\end{theorem}

\begin{definition}\label{def:gaussian_transform}
A Gaussian embedding $S$ is defined to be $\frac{1}{\sqrt{m}} \cdot G \in \mathbb{R}^{m\times n}$ with $m=\Theta(d/\varepsilon^2)$, where each entry of $G\in \mathbb{R}^{m\times n}$ is chosen independently from the standard Gaussian distribution.
\end{definition}

\begin{theorem}[\cite{woodruff2014sketching}]
For a given matrix $A \in \mathbb{R}^{n \times d}$ and $\varepsilon\in(0,1/2)$. 
Let $S\in\mathbb{R}^{\Theta(d/\varepsilon^2)\times n}$ be a Gaussian embedding. With probability at least $0.9999$, $S$ is an $\ell_2$ subspace embedding for the column space of $A$.
\end{theorem}

\begin{definition}\label{def:fast_gaussian_transform}
A composition of Gaussian embedding and \cous~embedding is defined to be $S' = S \Pi$, where $\Pi\in \mathbb{R}^{\Theta(d^2/\varepsilon^2) \times n}$ is a \cous~embedding and $S\in \mathbb{R}^{\Theta(d/\varepsilon^2) \times \Theta(d^2/\varepsilon^2)}$ is a Gaussian embedding.
\end{definition}

The following corollary directly follows from the above two theorems. 
\begin{corollary}\label{cor:fast_gaussian_transform}
For a given matrix $A \in \mathbb{R}^{n \times d}$ and $\varepsilon\in(0,1/2)$. 
Let $S'\in\mathbb{R}^{\Theta(d/\varepsilon^2) \times n}$ be a composition of Gaussian embedding and \cous~embedding. With probability at least $0.9998$, $S'$ is an $\ell_2$ subspace embedding for the column space of $A$.
Furthermore, $S'A$ can be computed in $O(\nnz(A) + d^4/\varepsilon^4)$ time.
\end{corollary}

We remark that all $\ell_2$ subspace embeddings introduced in this section are {\em oblivious}, meaning that the distribution of the embedding matrix does not depend on the matrix $A$.

\subsubsection{Properties of Symmetric Norms}\label{sec:prop_sym}

\paragraph{General Properties.} We first introduce several general properties of symmetric norms.
\begin{lemma}[Lemma 2.1 in~\cite{blasiok2017streaming}]\label{lem:sym_mono}
For any symmetric norm $\| \cdot \|_{\ell}$ and $x, y \in \mathbb{R}^n$ such that for all $i\in[n]$ we have $|x_i| \leq |y_i|$, then $\| x \|_{\ell} \leq \| y \|_{\ell}$.
\end{lemma}

\begin{lemma}[Fact 2.2 in~\cite{blasiok2017streaming}]\label{lem:sym_compare}
Suppose $\|\xi^{(1)}\|_{\ell}=1$, for any vector $x\in \R^n$, 
\[
\| x \|_{\infty} \leq \| x \|_{\ell} \leq \| x \|_1.
\]
\end{lemma}
\begin{lemma}[Lemma 3.12 in~\cite{blasiok2017streaming}]\label{lem:flat_median_lemma}
	Let $\|\cdot\|_{\ell}$ be a symmetric norm. Then
	\[
	\Omega(M_{\ell}/\sqrt{\log n})\leq \|\xi^{(n)}\|_{\ell}\leq O(M_{\ell}),
	\]
	where $M_{\ell}$ is as defined in Definition~\ref{def:median}.
\end{lemma}

\paragraph{Modulus of Approximation.}
We need the following quantity of a symmetric norm. 
\begin{definition}\label{def:mma}
	The {\em maximum modulus of approximation} of a symmetric norm $\|\cdot\|_{\ell}$ is defined as
	\begin{align*}
	\mma(\ell,r) = \max_{1\le a\le b\le a r \leq n}\frac{M_{\ell^{(a r)}}}{M_{\ell^{(b)}}},
	\end{align*}
	where $\|\cdot\|_{{\ell}^{(k)}}$ is a norm on $\R^k$ which is defined to be
	\[
	\| (x_1,x_2,\ldots,x_k) \|_{{\ell}^{(k)}} =  \| (x_1,x_2,\ldots,x_k,0,\ldots,0) \|_{\ell},
	\]
	and $M_{\ell^{(k)}}$ is as defined in Definition~\ref{def:median}.
\end{definition}

Intuitively, $\mma(\ell, r)$ characterizes how well the original symmetric norm can be approximated by a lower dimensional induced norm.
We show in the following lemma that $\mma(\ell, r)\le O(\sqrt{r \log n})$ for any symmetric norm.

\begin{lemma}\label{lem:mma_is_small}
For any symmetric norm $\|\cdot\|_{\ell}$ and $r\in [n]$,
$
\mma(\ell,r)\le  O(\sqrt{r\log n}).
$
\end{lemma}
\begin{proof}
By Lemma~\ref{lem:flat_median_lemma}, for any $i \in [n]$, $\Omega(M_{\ell^{(i)}}/\sqrt{\log n})\le \| \xi^{(i)} \|_{\ell} \le O( M_{\ell^{(i)}})$. 
Let $a r = c_1 b + c_2$, where $c_1, c_2$ are non-negative integers with $c_1\le a r/ b$ and $c_2\le b $.
Observe that we can rewrite $\xi^{(a r)}$ as
\begin{align*}
\xi^{(a r)} = \left[ \frac{\sqrt{b}}{\sqrt{a r}}\cdot \left(\underbrace{\xi^{(b)}, \xi^{(b)}, \xi^{(b)}, \ldots,\xi^{(b)}}_{c_1 \text{ times}}\right), \frac{\sqrt{c_2}}{\sqrt{a r}}\xi^{(c_2)}, 0, \ldots, 0 \right].
\end{align*}
Therefore, by triangle inequality, we have
\begin{align*}
\| \xi^{(a r)} \|_{\ell}
\le & ~ \frac{\sqrt{b}}{\sqrt{a r}} \cdot c_1\cdot \| \xi^{(b)}\|_{\ell} +  \frac{\sqrt{c_2}}{\sqrt{a r}} \cdot \| \xi^{(c_2)} \|_{\ell} \\
\le & ~ \sqrt{\frac{b }{a r}}\cdot {\frac{a r}{ b}} \cdot \| \xi^{(b)} \|_{\ell} + \frac{\sqrt{b}}{\sqrt{a r}} \cdot \| \xi^{(b)} \|_{\ell} \tag{$c_1 \le ar / b$ and $c_2 \le b$}\\
\le & ~ \sqrt{r} \cdot \| \xi^{(b)} \|_{\ell} +  \| \xi^{(b)} \|_{\ell} \tag{$a \le b \le ar$}\\
\le & ~ 2\sqrt{r} \cdot \| \xi^{(b)} \|_{\ell}. \tag{$r \ge 1$}
\end{align*}
Now we apply Lemma~\ref{lem:flat_median_lemma} on both sides, which implies
\begin{align*}
\frac{ M_{\ell^{( a r)}}}{\sqrt{\log n}}\le O(\sqrt{r}\cdot M_{\ell^{(b )}})
\end{align*}
as desired.
\end{proof}

\paragraph{Properties of \syms.} Now we introduce several properties of \syms.

The following lemma shows that for a data matrix $A \in \mathbb{R}^d$, calculating $SA$ requires $\wt{O}(\nnz(A) + \poly(d))$ time for a \syms~$S$.
\begin{lemma}\label{lem:time_of_S_dot_A}
	For a given matrix $A\in \R^{n\times d}$, let $S \in \R^{O(d) \times n}$ be a $\syms$ as in Definition~\ref{def:symmetric_sketch}.
	$S A$ can be computed in $O(\nnz(A))+\poly(d)$ time in expectation, and in $O(\nnz(A) \log n)+\poly(d)$ time in the worst case.
\end{lemma}
\begin{proof}
	Since $S$ is a $\syms,$ $S=\Pi\wt{D}=\Pi \cdot \left[ \begin{array}{c}
	w_0D_0 \\
	w_1D_1 \\
	\vdots\\
	w_tD_t
	\end{array}
	\right]	$, where $\Pi\in\mathbb{R}^{O(d) \times O(n \log n)}$.
	Since $D_i$ is a diagonal matrix, $\nnz(D_iA)\leq \nnz(A),$ and thus $\nnz(\wt{D}A)\leq (t + 1) \cdot \nnz(A)=O(\nnz(A)\log n)$, which implies $\wt{D}A$ can be computed in $(t + 1) \cdot \nnz(A)=O(\nnz(A)\log n)$ time. 

	On the other hand, the expected number of nonzero entries of $D_iA$ is $2^{-i}\nnz(A).$ Thus, $\wt{D}A$ has $O(\nnz(A))$ nonzero entries in expectation, which implies $\wt{D}A$ can be computed in $O(\nnz(A))$ time.
	
	Finally, notice that $\Pi$ is a composition of Gaussian embedding and \cous~embedding, which implies $\Pi\wt{D}A$ can be computed in $\nnz(\wt{D}A)+\poly(d)$ time. 
\end{proof}

The following lemma shows that with constant probability, for all $x\in\mathbb{R}^n,\|Sx\|_2 \le \poly(n)\|x\|_2$.
\begin{lemma}\label{lem:spectral_norm_S}
	 Let $S \in \R^{O(d) \times n}$ be a $\syms$ as defined in Definition~\ref{def:symmetric_sketch}, then with probability at least $0.9999$
	$
	\| S \|_2 \leq \poly(n).
	$
\end{lemma}
\begin{proof}
	Notice that $S=\Pi\wt{D},$ since $\|S\|_2 \le \|\Pi\|_2 \cdot \| \wt{D}\|_2$, it suffices to bound $\|\Pi\|_2$ and $\|\wt{D}\|_2.$ 
	Since $\Pi$ is a composition of Gaussian embedding and \cous~embedding (Definition~\ref{def:fast_gaussian_transform}), with probability at least $0.9999, \|\Pi\|_2 \le \|\Pi\|_F \le \poly(n).$ 
	Now consider $\wt{D}=\left[ \begin{array}{c}
	w_0D_0 \\
	w_1D_1 \\
	\vdots\\
	w_tD_t
	\end{array}
	\right]	$.
	By Lemma~\ref{lem:sym_compare}, for all $j\in[t], w_j \leq \poly(n)$.
	Furthermore, $\|D_j\|_2\leq 1$ and $t=\Theta(\log n)$, which implies $\|\wt{D}\|_2\leq \poly(n)$.
\end{proof}

Throughout this whole section we assume that for any non-zero vector $x \in \mathbb{R}^n$, we have $1 \leq |x_j| \leq \poly(n)$ for all $j \in [n]$.
Notice that this assumption is without loss of generality, as shown in the following lemma.
\begin{lemma}\label{lem:ignore}
For any non-zero vector $x \in \mathbb{R}^n$, let $\overline{x} \in \mathbb{R}^n$ be a vector with $\overline{x} = \frac{\poly(n) \cdot  x}{\|x\|_{\infty}}$, and $x' \in \mathbb{R}^n$ where
\[
x_i' = \begin{cases}
\overline{x}_i & \text{if $\overline{x}_i \ge 1$}\\
0 & \text{otherwise}\\
\end{cases}.
\]
For a symmetric norm $\|\cdot\|_{\ell}$, 
suppose $\|S\|_2 \le \poly(n)$ and \[\alpha \|x'\|_{\ell} \le \|Sx'\|_2 \le \beta \|x'\|_{\ell}\]
for some $\alpha, \beta \in [1 / \poly(n), \poly(n)]$, 
then \[\Omega(\alpha) \|x\|_{\ell} \le \|Sx\|_2 \le O(\beta) \|x\|_{\ell}.\]
\end{lemma}
\begin{proof}
By triangle inequality and Lemma~\ref{lem:sym_mono}, we have
\[
\|\overline{x}\|_{\ell}-n\leq \|x'\|_{\ell}\leq \|\overline{x}\|_{\ell}.
\] 
By Lemma~\ref{lem:sym_compare}, 
\[\|x'\|_{\ell}\geq \|x'\|_{\infty}=\|\overline{x}\|_{\infty}=\poly(n),\]
we have
\[
(1 - 1 / \poly(n))\|\overline{x}\|_{\ell}\leq \|x'\|_{\ell}\leq \|\overline{x}\|_{\ell}.
\]
Notice that $\|S \overline{x}\|_2=\|Sx'+S(\overline{x}-x')\|_2$. 
By triangle inequality we have
\[\|Sx'\|_2-\|S\|_2\|\overline{x}-x'\|_2\leq \|S\overline{x}\|_2\leq \|Sx'\|_2+\|S\|_2\|\overline{x}-x'\|_2.\] 
By the given conditions, we have
\[(1 - 1 / \poly(n))\|Sx'\|_2 \leq \|S\overline{x}\|_2\leq (1 + 1 / \poly(n))\|Sx'\|_2,\] 
which implies
\[\Omega(\alpha) \|\overline{x}\|_{\ell} \le \|S\overline{x}\|_2 \le O(\beta) \|\overline{x}\|_{\ell}.\]
Since $\overline{x} = \frac{\poly(n) \cdot  x}{\|x\|_{\infty}}$, we have
\[\Omega(\alpha) \|x\|_{\ell} \le \|Sx\|_2 \le O(\beta) \|x\|_{\ell}.\]
\end{proof}

By Lemma~\ref{lem:ignore}, we can focus on those non-zero vectors $x \in \mathbb{R}^n$ such that $1 \leq |x_j| \leq \poly(n)$ for all $j \in [n]$.

\begin{definition}\label{def:alg_contributing_level}
	For a given vector $x\in \R^n$, suppose for all $j \in [n]$, 
	\[
	1 \leq |x_j| \leq \poly(n).
	\]
	Let $g = \Theta(\log n)$. For each $i \in \{ 0,1,\ldots, g \}$, we define
	\[
	L_i (x) = \{ j ~ | ~ 2^i \leq |x_j| < 2^{i+1} \}.
	\]
	For each $i\in \{0,1,\ldots, g\}$, we define $V_i(x) \in \R^n$ to be the vector
	\begin{align*}
	V_i(x) = ( \underbrace{ 2^i, 2^i, \ldots, 2^i }_{|L_i(x)|}, 0, \ldots, 0).
	\end{align*}
	For each $i\in \{0,1,\ldots,g\}$, we say a level $i$ to be {\em contributing} if
	\[
	\| V_i(x) \|_{\ell} \geq \Omega\left(1/g\right) \cdot \| x \|_{\ell}.
	\]
\end{definition}

\begin{lemma} Let $g=\Theta(\log n)$. For a given vector $x\in\mathbb{R}^n$ such that for all $j\in[n],1\leq |x_j|\leq 2^g$, there exists at least one level $i\in\{0,1,\ldots,g\}$ which is contributing.
\end{lemma}
\begin{proof}
If none of $i\in\{0,1,\ldots,g\}$ is contributing, then $\|x\|_{\ell}\leq \sum_{i=0}^g \|V_i(x)\|_{\ell}\leq 1/(2g)\cdot \sum_{i=0}^g \|x\|_{\ell}\leq \frac{1}{2}\|x\|_{\ell}$, which leads to a contradiction.
\end{proof}

\subsection{Proof of Lemma~\ref{lem:example_mmc}}
\begin{proof}
Consider a fixed $n'\in [n]$.
By Lemma~\ref{lem:flat_median_lemma}, we have
\[
M_{\ell_a^{(n')}} = \Omega \left( \|\xi^{(n')}\|_{\ell_a^{(n')}} \right) = 
\Omega\left(1+c\sqrt{n'}\right)
\]
and
\[
M_{\ell_b^{(n')}} =  \Omega \left( \|\xi^{(n')}\|_{\ell_b^{(n')}}  \right) =
\Omega \left( \max\left(1,c\sqrt{n'}\right)\right).
\]
It is also straightforward to verify that 
\[
\max_{x \in  \mathbb{S}^{n'-1}} \|x\|_{\ell_a}
= 1 + c\sqrt{n'} 
\]
and
\[
\max_{x \in  \mathbb{S}^{n'-1}} \|x\|_{\ell_b}
= \max\left(1, c\sqrt{n'}\right).
\]
Taking the ratio between $\max_{x \in  \mathbb{S}^{n'-1}} \|x\|_{\ell}$ and $M_{\ell^{(n')}}$ for $\ell \in \{\ell_a, \ell_b\}$, we complete the proof.
\end{proof}
\subsection{Proof of Lemma~\ref{lem:orlicz_mmc}}
\begin{proof}
	Let $\overline{G}(x) =\sqrt{x}\cdot G^{-1}(1/x)$, where $G^{-1}(1 / x)$ is the unique value in $[0, \infty)$ such that $G(G^{-1}(1 / x)) = 1 / x$.
	We first show that $\overline{G}(x)$ is an approximately decreasing function for $x\in (0,\infty)$.
	Let  $m,n$ be two real numbers with $0< m\le n$.
	We have $1/n\le 1/m$, which implies
	$0<G^{-1}(1/n)\le G^{-1}(1/m)$ by monotonicity of $G$.
	By the third condition in Assumption~\ref{assump:property_P}, we have
	\[
	\frac{G(G^{-1}(1/m))}{G(G^{-1}(1/n))}
	\le C_{G} \cdot \bigg(\frac{G^{-1}(1/m)}{G^{-1}(1/n)}\bigg)^2,
	\]
	which implies
	\[
	\sqrt{n}\cdot G^{-1}(1/n)
	\le \sqrt{C_{G}} \cdot \sqrt{m}\cdot{G^{-1}(1/m)}.
	\]
	Hence $\overline{G}(n)\le \sqrt{C_G}\cdot \overline{G}(m)$.
	
	We are now ready to prove the lemma.
	Recall that for the Orlicz norm $\|\cdot\|_{\ell} = \|\cdot\|_G$, we have
	\[
	\mmc(\ell) =\max_{n'\in [n]}
	\mc({\ell}^{(n')}) = \max_{n'\in [n]} \frac{\max_{x\in \mathbb{S}^{n'-1}}\|x\|_{{\ell}^{(n')}}}{M_{{\ell}^{(n')}}}.
	\]
	By Lemma~\ref{lem:flat_median_lemma}, we have,
	\[
	\Omega(1)\cdot  \|\xi^{(n')}\|_{{\ell}^{(n')}}\le 
	M_{\ell^{(n')}}
	\le O(\sqrt{\log n})\cdot \|\xi^{(n')}\|_{{\ell}^{(n')}}.
	\]
	Thus $\|\xi^{(n')}\|_{{\ell}^{(n')}}$ provides an approximation to $M_{\ell^{(n')}}$.
	By definition of $\|\cdot\|_G$, we have
	\[
	\|\xi^{(n')}\|_{{\ell}^{(n')}} = \frac{1}{\sqrt{n'}\cdot G^{-1}(1/n')} = \frac{1}{\overline{G}(n')}.
	\]
	Hence
	\[
	\Omega(1) \le 
	M_{\ell^{(n')}}
	\cdot \overline{G}(n') \le O(\sqrt{\log n}).
	\]
	
	Next, we compute $\max_{x\in \mathbb{S}^{n'-1}} \|x\|_{{\ell}^{(n')}}$. 
	For an arbitrary unit vector $x\in \mathbb{S}^{n'-1}$, 
	we denote
	\[
	B_j  =\{i\in [n]: |x_i| \in [1/2^{j}, 1/2^{j-1})\} 
	\]
	and $b_j = |B_j|$.
	For each $j$, let $x^{B_j} \in \mathbb{R}^n$ be the vector such that
	\[x^{B_j}_{i} = \begin{cases}
	x_i & \text{if $j \in B_j$} \\
	0 & \text{otherwise}
	\end{cases}.
	\]
	Note that non-zero coordinates in $x^{B_{j}}$ have magnitude close to each other (within a factor of $2$), we thus have
	\[
	\|x^{B_{j}}\|_{\ell} = {\|x^{B_{j}}\|_2} \cdot 
	\left\| \frac{x^{B_{j}}}{\|x^{B_{j}}\|_2} \right \|_{\ell}
	\le \frac{\sqrt{b_j}}{2^{j-1}}\cdot \left\| \frac{x^{B_{j}}}{\|x^{B_{j}}\|_2} \right \|_{\ell}
	\le \frac{\sqrt{b_j}}{2^{j-2}}\cdot \|\xi^{(b_{j})}\|_{\ell^{(b_{j})}}
	= \frac{\sqrt{b_j}}{2^{j-2}}\cdot \frac{1}{\overline{G}(b_{j})}.
	\]
	Similarly,
	\[
	\|x^{B_{j}}\|_{\ell} \ge \frac{\sqrt{b_j}}{2^{j+2}}\cdot \|\xi^{(b_{j})}\|_{\ell^{(b_{j})}}
	= \frac{\sqrt{b_j}}{2^{j+2}}\cdot\frac{1}{\overline{G}(b_j)}
	\ge \frac{\sqrt{b_j}}{2^{j+2}}\cdot \frac{1}{\sqrt{C_G}\cdot \overline{G}(1)}.
	\]
	
	We claim there exists an constant $c>0$ such that 
	\[
	\sum_{\substack{j>c\log n \\ b_j > 0}} \|x^{B_{j}}\|_{\ell} \le 
	\sum_{j'\le c\log n} \|x^{B^{j'}}\|_{\ell}.
	\]
	To show this, by Lemma~\ref{lem:fasterthanlinear}, for any $b\ge 1$, 
	\[
	b = \frac{G(G^{-1}(1))}{G(G^{-1}(1/b))} \ge \frac{G^{-1}(1)}{G^{-1}(1/b)},
	\]
	which implies
	\[
	G^{-1}(1/b) \ge \frac{G^{-1}(1)}{b}.
	\]
	Next, since $\|x\|_2= 1$, there exists an $0\le \tilde{j}\le 4\log n$ such that $b_{\tilde{j}}\ge 1$.
	Therefore,
	\[
	\|x^{B_{\tilde{j}}} \|_{\ell}\ge \frac{1}{2^{\tilde{j}+2}}\cdot \frac{1}{\sqrt{C_G}\cdot \overline{G}(1)}.
	\]
	Thus, we have
	\begin{align*}
	\sum_{\substack{j>c\log n \\ b_j > 0}} \|x^{B_{j}}\|_{\ell} &=
	\sum_{\substack{j>c\log n \\ b_j > 0}}\|x^{B_j}\|_2\cdot \left\| \frac{x^{B_j}}{\|x^{B_j}\|_2} \right\|_{\ell}
	\le \sum_{\substack{j>c\log n \\ b_j > 0}}  \frac{\sqrt{n}}{2^{j-2}}\cdot \|\xi^{(b_{j})}\|_{\ell^{(b_{j})}}\\
	&\le \sum_{\substack{j>c\log n \\ b_j > 0}} \frac{\sqrt{n}}{2^{j-2}}\cdot \frac{1}{\sqrt{b_{j}}G^{-1}(1/b_{j})}
	\le n\cdot \frac{\sqrt{n}}{2^{c\log n-2}}\cdot \frac{b_j}{\sqrt{b_j}G^{-1}(1)}\\
	&\le n\cdot 2^{\tilde{j}+2} \cdot \frac{\sqrt{n}}{2^{c\log n-2}}\cdot \sqrt{C_G n}\cdot \|x^{B_{\tilde{j}}}\|_{\ell}
	\le  \|x^{B_{\tilde{j}}}\|_{\ell}\le \sum_{j'\le c\log n} \|x^{B^{j'}}\|_{\ell}
	\end{align*}
	for some sufficiently large constant $c$.
	
	Let
	\[
	j^*  = \argmax_{j\le c \log n}\|x^{B_{j}}\|_{\ell},
	\]
	we have
	\begin{align*}
	\|x^{B_{j^*}}\|_{\ell}\le \|x\|_{\ell} \le \sum_{j \le c\log n} \|x^{B_{j}}\|_{\ell}
	+ \sum_{\substack{j>c\log n \\ b_j > 0}} \|x^{B_{j}}\|_{\ell}
	\le 2\sum_{j \le c\log n} \|x^{B_{j}}\|_{\ell}
	\le O(\log n)\cdot\|x^{B_{j^*}}\|_{\ell}.
	\end{align*}
	Thus,
	\begin{align*}
	&\max_{x\in \mathbb{S}^{n'-1}} \|x\|_{\ell^{(n')}}
	\le  O(\log n') 	\max_{x\in \mathbb{S}^{n'-1}}  \|x^{B_{j^*}}\|_{\ell}
	\le  O(\log n') 	\max_{x\in \mathbb{S}^{n'-1}}  \|x^{B_{j^*}}\|_2  \cdot  \left\| \frac{x^{B_{j^*}}}{\|x^{B_{j^*}}\|_2} \right \|_{\ell}\\
	\le & O(\log n') \max_{x\in \mathbb{S}^{n'-1}}   \left\| \frac{x^{B_{j}}}{\|x^{B_{j}}\|_2} \right \|_{\ell}
	\le O(\log n')  \max_{b_{j*} \le n'}\|\xi^{(b_{j})}\|_{\ell^{(b_{j})}}
	\le O(\log n')  \max_{b_{j*} \le n'}  \frac{1}{\overline{G}(b_{j^*})}
	\le  \frac{ O(\sqrt{C_G}  \log n')}{ \overline{G}(n')}.
	\end{align*}

	Thus, we have
	\[
	\mmc(\ell)=\max_{n'\in [n]} 
	\frac{\max_{x\in \mathbb{S}^{n'-1}}\|x\|_{\ell^{(n')}}}{M_{\ell^{(n')}}}
	\le O(\sqrt{C_G} \log n). 
	\]
\end{proof}

\subsection{Contraction Bound of~\syms}
In this section we give the contraction bound of \syms.
We first show that for a fixed vector $x \in \mathbb{R}^n$, $\| \wt{D} x \|_2 \geq 1 / \poly(d \log n) \cdot \| x \|_{\ell}$ with probability $1- 2^{-\Theta(d \log n)}$.

\begin{lemma}\label{lem:input_sparsity_fixed_vector_no_contraction}
	Let $\wt{D}$ be the matrix defined in Definition~\ref{def:symmetric_sketch}. 
	For any fixed $x \in \mathbb{R}^n$, with probability $1- 2^{-\Theta(d \log n)}$,
	$
	\| \wt{D} x \|_2 \geq 1/\alpha_0\cdot \| x \|_{\ell}
	$, where $\alpha_0 = O(\mma(\ell, d ) \cdot \log^{5/2} n) =O(\sqrt{d}\log^3 n)$.
\end{lemma}

\begin{proof}
	The lemma follows from the following two claims.
	Recall that $t=\Theta(\log n)$.
	\begin{claim}\label{cla:input_sparsity_j_not_0}
		For any fixed $x\in\mathbb{R}^n$. If there is a contributing level $i^*\in \{0,1,2,\ldots, g\}$ such that $|L_{i^*}(x)|=  \Theta(d\log n) \cdot 2^j$ for some $j\in [t]$, then with probability at least $1-2^{-\Theta(d\log n)}$, 
		\begin{align*}
		\|w_jD_jx\|_2\geq \Omega \left( \frac{ 1 }{ \mma(\ell, d ) \cdot \log^{5/2} n } \right) \|x\|_{\ell}.
		\end{align*}
	\end{claim}
	
	\begin{proof}
		Let $y_h$ be a random variable such that 
		\[
			y_h=\begin{cases}
			1 & \text{if the $h$-th diagonal entry of $D_j$ is $1$} \\
			0 & \text{otherwise}
			\end{cases}.
		\]
 Let $Y=\sum_{h\in L_{i^*}(x)} y_h.$ 
		By Chernoff bound, we have
		\begin{align*}
		\Pr[ Y \geq \Omega(d\log n) ] \geq 1-2^{-\Theta(d\log n)}.
		\end{align*}
		Conditioned on $Y\geq \Omega(d\log n)$,
		we have
		\begin{align*}
		\frac{\|w_jD_jx\|_2}{\|x\|_{\ell}} = & ~ \frac{w_j\|D_jx\|_2}{\|x\|_{\ell}}\\
		\geq & ~ \frac{2^{i^*} w_j \sqrt{d\log n}}{\|x\|_{\ell}}\\
		\geq & ~ \frac{2^{i^*} \sqrt{2^j} M_{\ell^{(2^j)}} \sqrt{d\log n}}{2^{i^*+1}g M_{\ell^{(|L_{i^*}(x)|)}}\sqrt{\log n}\sqrt{|L_{i^*}(x)|}}\\
		\geq & ~\Omega(1/\log^{3/2} n)\cdot \frac{M_{\ell^{(2^j)}}}{M_{\ell^{(|L_{i^*}(x)|)}}}\\
		= & ~\Omega(1/\log^{3/2} n)\cdot \frac{M_{\ell^{(2^j)}}}{M_{\ell^{(|L_{i^*}(x)|/\log n)}}} \cdot \frac{M_{\ell^{(|L_{i^*}(x)|/\log n)}}}{M_{\ell^{(|L_{i^*}(x)|)}}}\\
		\geq & ~ \Omega\left(\frac{1}{\mma(\ell, d)\cdot\mma(\ell,\log n)\cdot\log^{3/2} n}\right)\\
		\geq &~ \Omega\left(\frac{1}{\mma(\ell, d)\cdot\log^{5/2} n}\right).\\
		\end{align*}
		Here the first inequality follows from the fact that there are at least $\Omega(d\log n)$ coordinates sampled from $L_{i^*}(x)$.
		The second inequality follows from Lemma~\ref{lem:flat_median_lemma} and the fact that level $i^*$ is a contributing level. The third inequality follows from $|L_{i^*}(x)|=  \Theta(d\log n) \cdot 2^j$ and $g=\Theta(\log n).$ The forth inequality follows from  Definition~\ref{def:mma}. The last inequality follows from Lemma~\ref{lem:mma_is_small}.
	\end{proof}
	
	\begin{claim}\label{cla:input_sparsity_j_is_0}
		For any fixed $x\in\mathbb{R}^n$. If there is a contributing level $i^*\in \{0,1,2,\ldots, g\}$ such that $|L_{i^*}(x)|=  O(d\log n)$, then we have
		\begin{align*}
		\|w_0D_0x\|_2\geq  \Omega \left( \frac{ 1 }{ \mma(\ell, d) \cdot \log^{5/2} n } \right) \|x\|_{\ell}.
		\end{align*}
	\end{claim}
	
	\begin{proof}
		\begin{align*}
		\frac{\|w_0 D_0 x\|_2}{\|x\|_{\ell}}  = & ~ \frac{w_0\|x\|_2}{\|x\|_{\ell}}\\
		\geq & ~ \frac{2^{i^*} w_0\sqrt{|L_{i^*}(x)|}}{\|x\|_{\ell}}\\
		\geq & ~ \frac{2^{i^*} M_{\ell^{(1)}}\sqrt{|L_{i^*}(x)|}}{2^{i^*+1}g M_{\ell^{(|L_{i^*}(x)|)}}\sqrt{\log n}\sqrt{|L_{i^*}(x)|}}\\
		\geq & ~ \Omega(1/\log^{3/2} n)\cdot \frac{M_{\ell^{(1)}}}{M_{\ell^{(|L_{i^*}(x)|/\log n)}}}\cdot \frac{M_{\ell^{(|L_{i^*}(x)|/\log n)}}}{M_{\ell^{(|L_{i^*}(x)|)}}}\\
		\geq & ~ \Omega \left( \frac{1}{\mma(\ell, d)\cdot\log^{5/2} n}\right)
		\end{align*}
		The first inequality follows from the fact that we only consider the contribution of the coordinates in $L_{i^*}(x)$. The second inequality follows from Lemma~\ref{lem:flat_median_lemma} and the fact that level $i^*$ is a contributing level. The third inequality follows from $g=\Theta(\log n).$
		The last inequality follows from Definition~\ref{def:mma} and Lemma~\ref{lem:mma_is_small}.
	\end{proof}
	By Claim~\ref{cla:input_sparsity_j_not_0} and Claim~\ref{cla:input_sparsity_j_is_0}, since any vector $x\in\mathbb{R}^n$ contains at least one contributing level, with probability at least $1-2^{-\Theta(d\log n)}$ we have $\|\wt{D}x\|_2\geq \Omega(1/(\mma(\ell,d)\cdot \log^{5/2} n))\cdot \|x\|_{\ell}$.
	We complete the proof by combining this with Lemma~\ref{lem:mma_is_small}.
\end{proof}

Now we show how to combine the contraction bound in Lemma~\ref{lem:input_sparsity_fixed_vector_no_contraction} with a net argument to give a contraction bound for all vectors in a subspace.

\begin{lemma}\label{lem:no_contraction_for_all}
Let $S\in\mathbb{R}^{O(d)\times n}$ be a random matrix.
For any $\alpha_0 =\poly(n)$ and $A\in\mathbb{R}^{n\times d}$, if
\begin{enumerate}
\item $\|S\|_2\leq \poly(n)$ holds with probability at least $0.999$;
\item for any fixed $x\in\mathbb{R}^n,$ $\|Sx\|_2\geq 1/\alpha_0\cdot\|x\|_{\ell}$ holds with probability $1-e^{-C d\log n}$ for a sufficiently large constant $C$,
\end{enumerate}
then with probability at least $0.998$, for all $y \in \mathbb{R}^n$ in the column space of $A$,
\[\|Sy\|_2\geq \Omega(1/\alpha_0 ) \|y\|_{\ell}.\]
\end{lemma}
\begin{proof}
For the matrix $A\in \R^{n\times d}$, we define the set ${\cal B} = \{ y ~|~ y = A x , \| y \|_2 = 1\}$. We define ${\cal N} \subset \R^n$ to be an $\epsilon$-net of $\mathcal{B}$ as in Definition~\ref{def:net}.
By Lemma~\ref{lem:net_size}, we have
$|{\cal N}| \leq (1+ 1/\epsilon)^d$,
and for all $y\in {\cal B}$, there exists $z \in {\cal N}$ such that $\| y -z \|_2 \leq \epsilon$.
We take $\varepsilon = 1 / \poly(n)$ here.

Due to the second condition, since $|\mathcal{N}| \le e^{O(d \log n)}$, by taking union bound over all vectors in $\mathcal{N}$, we know that with probability $1-e^{-\Theta(d\log n)}$, for all $z \in {\cal N}$,
$
\| S z \|_2 \geq 1/\alpha_0 \cdot \| z \|_{\ell}.
$

Now, for any vector $y\in {\cal B}$, there exists $z \in {\cal N}$ such that $\| y -z \|_2 \leq 1 / \poly(n)$, and we define $w= y-z$.
\begin{align*}
\| S y \|_2 = & ~ \| S (z + w ) \|_2 \\
\geq & ~ \| S z \|_2 - \|S w\|_2 & \tag{triangle~inequality}\\
\geq & ~ 1/\alpha_0 \cdot \| z \|_{\ell} - \| S w\|_2 & \tag{by the second condition} \\
\geq & ~ 1/\alpha_0 \cdot\| z \|_{\ell} - \| S \|_2 \| w\|_2 & \tag{$\| S w\|_2 \leq \|S\|_2 \cdot \| w \|_2$}\\
\geq & ~ 1/\alpha_0 \cdot\| z \|_{\ell} - \poly(n) \cdot \| w \|_2 &\tag{by the first condition} \\
\geq & ~ 1/\alpha_0 \cdot\| y-w \|_{\ell} - \poly(n) \cdot \| w \|_2 & \tag{$y = z+w$} \\
\geq & ~ 1/\alpha_0 \cdot\| y \|_{\ell} - 1/\alpha_0\cdot \| w \|_{\ell} - \poly(n) \cdot \| w \|_2 & \tag{triangle inequality} \\
\geq & ~ 1/\alpha_0 \cdot\| y \|_{\ell} - 1/\alpha_0\cdot \sqrt{n} \| w \|_2 - \poly(n) \cdot \| w \|_2 &  \tag{Lemma~\ref{lem:sym_compare}} \\
\geq & ~ 1/\alpha_0 \cdot\| y \|_{\ell} - ( 1/\alpha_0\cdot \sqrt{n}  + \poly(n) ) \epsilon &  \tag{$\| w \|_2 \leq \epsilon$} \\
\geq & ~ 0.5 /\alpha_0 \cdot\| y \|_{\ell}.
 \end{align*}
\end{proof}
\begin{lemma}\label{lem:contraction_for_all_symsketch}
For a given  matrix $A\in\mathbb{R}^{n\times d}$.
Let $S\in\mathbb{R}^{O(d)\times n}$ be a \syms~as defined in Definition~\ref{def:symmetric_sketch}.
With probability at least $0.995$, for all $x\in\mathbb{R}^d$, $\|SAx\|_2\geq 1/\alpha_0\cdot \|Ax\|_{\ell}$ where $\alpha_0=O(\sqrt{d}\log^3 n)$.
\end{lemma}

\begin{proof}
By Lemma~\ref{lem:input_sparsity_fixed_vector_no_contraction} and Lemma~\ref{lem:spectral_norm_S}, the two conditions in Lemma~\ref{lem:no_contraction_for_all} are satisfied. 
By Lemma~\ref{lem:no_contraction_for_all}, with probability at least $0.998$, for all $x\in\mathbb{R}^d,$ $\|\tilde{D}Ax\|_2\geq \Omega( 1/\alpha_0)\|Ax\|_{\ell}$. Since $\Pi \in \R^{O(d) \times  n(t+1)}$ is a composition of Gaussian embedding and \cous~embedding with $\varepsilon = 0.1$, by Corollary~\ref{cor:fast_gaussian_transform}, with probability at least $0.999$, for all $x\in\mathbb{R}^d$, $\|\Pi\tilde{D}Ax\|_2\geq \Omega(\|\tilde{D}Ax\|_2)$. By a union bound, we know that with probability at least $0.995$, for all $x\in\mathbb{R}^d$, $\|SAx\|_2\geq \Omega(1/\alpha_0) \|Ax\|_{\ell}$.
\end{proof}

\subsection{Dilation Bound of \syms}\label{sec:alg_no_dilation_for_each}
In this section we give the dilation bound of \syms.
We first show that for any fixed $x\in\mathbb{R}^n,$ with high probability, $\|\tilde{D}x\|_2 \le \poly(d \log n) \cdot \mmc(\ell)  \cdot \|x\|_{\ell}$.

\begin{lemma}\label{lem:linear_regression_no_dilation_detailed}
Let $\wt{D}$ be the matrix defined in Definition~\ref{def:symmetric_sketch}. 
	For any fixed vector $x\in\mathbb{R}^n$, with probability $1 - \delta$, $\|\tilde{D}x\|_2 \le \alpha_1 / \delta \cdot \|x\|_{\ell}$, where $\alpha_1 = O(\mmc(\ell) \log^{5 /2} n)$.
	
\end{lemma}
\begin{proof}
Consider a fixed vector $x\in\mathbb{R}^n$.
Recall that $t=\Theta(\log n)$. 
Let $c>0$ be a fixed constant.
We define the $j$-heavy level set $H_j$ as
\begin{align*}
H_j = \left\{ i ~\bigg|~ |L_i(x)| \geq c \frac{\delta2^j}{  \log^2 n}, 0 \leq i\leq g \right\}.
\end{align*}
Let $\ov{H}_j$ be the $j$-light level set, i.e.,
$
\ov{H}_j = \{ 0,1,\ldots, g\} \backslash H_j.
$
Notice that \[\sum_{i\in\ov{H}_j} |L_i(x)|\cdot 2^{-j}\leq g\cdot c \delta 2^j/\log^2 n\cdot 2^{-j}\leq O(\delta/\log n).\]
By Markov's inequality,
with probability at least $1-\delta/(2t)$, no element from a $j$-light level is sampled by $D_j$, i.e., for all $i\in \ov{H}_j,k\in L_i(x)$, the $k$-th diagonal entry of $D_j$ is $0$. By taking union bound over all $j\in [t]$, with probability at least $1-\delta/2$, for all $j\in[t]$, no element from a $j$-light level is sampled by $D_j$. Let $\zeta$ denote this event.
We condition on this event in the remaining part of the proof. 
In the following analysis, we show an upper bound of $\|w_jD_jx\|_2^2$ for each $j\in[t]$.
Let $H_j$ be the set of $j$-heavy levels.

Consider a fixed $j\in[t]$, we have
\begin{align*}
\E_{D_j } \left[ \left\|w_j  D_j x \right\|_2^2 ~\bigg|~ \zeta \right]
= & ~ w_j^2 \E_{D_j} \left[ \left\| D_j x \right\|_2^2 \bigg| \zeta \right] \\
= & ~ w_j^2 \E_{D_j} \left[ \sum_{h=1}^n  ( D_j(h,h) )^2 x_h^2 ~ \bigg| ~ \zeta \right] \\
= & ~ w_j^2 \E_{D_j} \left[ \sum_{i=0}^g \sum_{h\in L_i(x)}  ( D_j(h,h) )^2 x_h^2 ~ \bigg| ~ \zeta \right] \\
= & ~ w_j^2 \frac{1}{2^j} \sum_{i\in H_j} \sum_{h\in L_i(x)} x_h^2 \\
\leq & ~ w_j^2 \frac{1}{2^j} \sum_{i\in H_j} |L_i(x)| \cdot ( 2^{i+1} )^2.
\end{align*}
\begin{claim}\label{cla:wj2j_Ml2j}
$w_j^2 2^{-j} \leq O ( (M_{\ell^{(2^j)}})^2 )$ .
\end{claim}
\begin{proof}
\begin{align*}
w_j^2 2^{-j} = & ~ ( \| (1,1,\ldots, 1, 0, \ldots, 0)  \|_{\ell} )^2 \cdot 2^{-j} \\
= & ~ \left( \left\| \frac{1}{ \sqrt{2^j} } (1,1,\ldots,1,0,\ldots, 0) \right\|_{\ell} \sqrt{2^{j}} \right)^2 \cdot 2^{-j} \\
= & ~ ( \| \xi^{(2^j)} \|_{\ell} )^2 \cdot 2^j \cdot 2^{-j} \\
= & ~ ( \| \xi^{(2^j)} \|_{\ell} )^2 \\
\leq & ~ O(M_{\ell^{(2^j)}}),
\end{align*}
where the third step follows from the definition of $\xi^{(2^j)}$, and the last step follows from Lemma~\ref{lem:flat_median_lemma}.
\end{proof}

Using the above claim, we have
\begin{align*}
 & ~ w_j^2\sum_{i\in H_j}|L_i(x)|\cdot 2^{2i-j}\\
= & ~ \sum_{i\in H_j}\frac{w_j^22^{-j}}{   (M_{\ell^{(|L_i(x)|)}})^2}\cdot |L_i(x)|\cdot  (M_{\ell^{(|L_i(x)|)}})^2\cdot 2^{2i}\\
\leq & ~O\left(\sum_{i\in H_j}  \left(\frac{M_{\ell^{(2^j)}}}{M_{\ell^{(|L_i(x)|)}}}\right)^2\cdot|L_i(x)|\cdot   (M_{\ell^{(|L_i(x)|)}})^2\cdot 2^{2i}\right)\\
\leq & ~O\left(\sum_{i\in H_j}  \left(\frac{M_{\ell^{(2^j)}}}{M_{\ell^{(|L_i(x)|)}}}\right)^2 w_{\log |L_i(x)| }^2 \log n \cdot 2^{2i}\right) \\
= & ~ \log n\cdot\underbrace{ \sum_{i\in H_j , |L_i(x)|\leq 2^j} \left(\frac{M_{\ell^{(2^j)}}}{M_{\ell^{(|L_i(x)|)}}}\right)^2\cdot w_{\log |L_i(x)| }^2 \cdot 2^{2i} }_{\diamondsuit}\\
 & ~+ \log n\cdot\underbrace{ \sum_{i\in H_j , |L_i(x)| > 2^j} \left(\frac{M_{\ell^{(2^j)}}}{M_{\ell^{(|L_i(x)|)}}}\right)^2\cdot  w_{\log |L_i(x)| }^2 \cdot 2^{2i} }_{\heartsuit},
\end{align*}
where the second step follows from $w_j^2 2^{-j} \leq O( ( M_{\ell^{(2^j)}})^2)$ (Claim~\ref{cla:wj2j_Ml2j}), and the third step follows from $ |L_i(x)| \cdot (M_{\ell^{(|L_i(x)|)}})^2 \leq O(w_{\log |L_i(x)| }^2\log n )$ (Lemma~\ref{lem:flat_median_lemma}). It remains to upper bound $\diamondsuit$ and $\heartsuit$.

To given an upper bound for $\diamondsuit$, we have
\begin{align*}
\diamondsuit \leq & ~O\left( \sum_{i\in H_j , |L_i(x)|\leq 2^j} \mma^2(\ell,  \log^2 n/\delta) \cdot w_{\log|L_i(x)|}^2\cdot 2^{2i} \right) \\
\leq & ~O\left( \mma^2(\ell, \log^2 n/\delta) \left( \sum_{i=0}^{g} w_{\log |L_i(x)|}\cdot 2^{i} \right)^2 \right)\\
\leq &~ O\left( \mma^2(\ell,  \log^2 n/\delta)\right) \| x \|_{\ell}^2\\
\leq &~O(\log^3 n/\delta)\|x\|_{\ell}^2,
\end{align*}
where the first step follows from the definition of $\mma$, the second step follows from Minkowski inequality, the third step follows from the definition of $L_i(x)$, $w_{\log|L_i(x)|}$ and triangle inequality, the last step follows from Lemma~\ref{lem:mma_is_small}.

To give an upper bound for $\heartsuit$, we have
\begin{align*}
\heartsuit \leq & ~O\left(\log n\cdot \sum_{i\in H_j , |L_i(x)| > 2^j}  \mmc^2(\ell) w_{\log|L_i(x)|}^2\cdot 2^{2i}\right)\\
\leq & ~ O\left(  \log n\cdot \mmc^2(\ell)\cdot \left( \sum_{i=0}^{g} w_{\log |L_i(x)|}\cdot 2^{i} \right)^2\right)\\
\leq & ~O\left( \log n\cdot \mmc^2(\ell)\cdot \|x\|_{\ell}^2\right),
\end{align*}
where the first step follows from $( M_{\ell^{(2^j)}} / M_{\ell^{(|L_i(x)|)}} )^2 \leq O( \log n  \cdot \mmc^2(\ell)) $ (Lemma 3.14 in \cite{blasiok2017streaming}).

Putting it all together, we have
\begin{align*}
\E_{D_j } [ \|w_j  D_j x\|_2^2 | \zeta ] \leq \log n \cdot ( \diamondsuit + \heartsuit) \leq   O(\log^4 n/\delta + \log^2 n\cdot \mmc^2(\ell)) \| x \|_{\ell}^2.
\end{align*}
Thus, 
\[
\E_{\tilde{D}}[\|\tilde{D}x\|_2^2 | \zeta] \leq \sum_{j=0}^t \E_{D_j } [ \|w_j  D_j x\|_2^2 | \zeta ] \leq O(\log^5 n/\delta + \log^3 n\cdot \mmc^2(\ell)) \| x \|_{\ell}^2.
\]
By Markov's inequality, conditioned on $\zeta$, with probability at least $1-\delta/2$, 
\[\|\tilde{D}x\|_2^2\leq O(\log^5 n /\delta + \log^3 n\cdot \mmc^2(\ell)) \| x \|_{\ell}^2/\delta.\]
Since $\zeta$ holds with probability at least $1-\delta/2$, with probability at least $1-\delta$, we have 
\[\|\tilde{D}x\|_2\leq O(\log^{5/2} n /\delta\cdot \mmc(\ell))\cdot \|x\|_{\ell}.\]
\end{proof}

Now we show how to use the dilation bound for a fixed vector in Lemma~\ref{lem:linear_regression_no_dilation_detailed} to prove a dilation bound for all vectors in a subspace.
We need the following existential result in our proof. 

\begin{lemma}[\cite{auerbach1930area}]\label{lem:existence_well_conditioned_basis}
	Given a matrix $A\in \R^{n\times m}$ and a norm $\|\cdot \|$, there exists a basis matrix $U\in \R^{n\times d}$ of the column space of $A$, such that 
	\[
	\sum_{i=1}^d \| U^i \| \leq d, 
	\]
	and for all $x \in \R^d$,
	\[
	\| x \|_{\infty} \leq \| U x \|.
	\]

\end{lemma}

\begin{lemma}\label{lem:no_dilation_for_all_multiple_copies}
	Given a matrix $A\in\R^{n\times d}$. 
	Let $S\in\mathbb{R}^{O(d)\times n}$ be a \syms~as defined in Definition~\ref{def:symmetric_sketch}.
	With probability at least $0.99$, for all $x\in \R^d$,
	\begin{align*}
	\| S A x \|_2 \leq O(\alpha_1 d^2) \| A x \|_{\ell},
	\end{align*}
	where $\alpha_1 = O(\mmc(\ell)\cdot \log^{5/2} n)$.
\end{lemma}
\begin{proof}
	Recall that $S=\Pi\tilde{D}$.
	Let $U$ be a basis matrix of the column space of $A$ as in Lemma \ref{lem:existence_well_conditioned_basis}.
	By Lemma~\ref{lem:linear_regression_no_dilation_detailed}, for a fixed $i\in [d]$, with probability at least $1 - 1/(100d)$,
	$
	\| \tilde{D} U^i \|_2 \leq O(\alpha_1 d) \| U^i \|_{\ell}.
	$
	By taking a union bound over $i\in [d]$, with probability at least $0.999$, for all $i \in [d]$,
	$
	\| \tilde{D}U^i\|_2 \leq \alpha_1 d  \| U^i \|_{\ell}.
	$
	Thus, for any $x\in \R^d$,
	\begin{align*}
	\| \tilde{D} U x \|_2 \leq & ~ \sum_{i=1}^d |x_i| \cdot \| \tilde{D} U^i \|_2 \\
	\leq & ~ \| x \|_{\infty} \cdot \sum_{i=1}^d \| \tilde{D} U^i \|_2 \\
	\leq & ~  \| U x \|_{\ell} \cdot \sum_{i=1}^d \| \tilde{D} U^i \|_2 \\
	\leq & ~ O(\alpha_1 d)  \cdot \| U x \|_{\ell} \cdot \sum_{i=1}^d \| U^i \|_{\ell} \\
	\leq & ~  O(\alpha_1 d^2) \cdot \| U x \|_{\ell},
	\end{align*}
	where the first step follows from triangle inequality, the second step follows from $|x_i|\leq \| x \|_{\infty}$ for all $i \in [d]$, the third step follows from $\| x \|_{\infty} \leq \| U x \|_{\ell} $, the fourth step follows from $ \| \tilde{D} U^i \|_2 \leq  O(\alpha_1 d) \| U^i \|_{\ell}$, the last step follows from $\sum_{i=1}^d \| U^i \|_{\ell} \leq d$.
	
	By Corollary~\ref{cor:fast_gaussian_transform}, with probability at least $0.999$, $\Pi$ is an $\ell_2$ subspace embedding with $\varepsilon = 0.1$ for the column space of $\tilde{D}U$. 
Thus, with probability at least $0.99$, for all $x\in\mathbb{R}^d$, $\|SAx\|_2\leq O(\alpha_1 d^2)\|Ax\|_{\ell}$. 
\end{proof}

\subsection{Proof of Theorem~\ref{thm:subspace_embeddings_formal}}

\begin{proof}
	It directly follows from Lemma~\ref{lem:contraction_for_all_symsketch}, Lemma~\ref{lem:no_dilation_for_all_multiple_copies} and Lemma~\ref{lem:time_of_S_dot_A}.
\end{proof}

%% file: main_appendix.tex
\section{Missing Proofs of Main Theorems}

\subsection{Proof of Theorem~\ref{thm:orlicz}} \label{sec:proof_thm1}
Let $S\in\mathbb{R}^{O(d)\times n}$ be a \syms~as defined in Definition~\ref{def:symmetric_sketch}, and $\Pi = O(\sqrt{d }\log^3 n) \cdot S$.
By Corollary~\ref{coro:orlicz_ose}, for a given matrix $A\in \R^{n\times d}$, with probability at least $0.9$, for all $x \in \mathbb{R}^d$,
\[
\|A x\|_G \le \|\Pi A  x\|_2 \le \kappa \| A x\|_G,
\]
where $\kappa = O(\sqrt{C_G}d^{5/2} \log^{13/2} n)$.
We prove Theorem~\ref{thm:orlicz} by combining Theorem~\ref{thm:reg_with_se} with the embedding matrix $\Pi$ constructed above.

\subsection{Proof of Theorem~\ref{thm:sym}} \label{sec:proof_thm2}
Let $S\in\mathbb{R}^{O(d)\times n}$ be a \syms~as defined in Definition~\ref{def:symmetric_sketch}.
For a given data matrix $A \in \mathbb{R}^{n \times d}$ and response vector $b \in \mathbb{R}^n$, we calculate $x^* = \argmin_x \|SAx - Sb\|_2$ and return $x^*$.
The algorithm runs in $O(\nnz(A) + \poly(d))$ time, since by Lemma~\ref{lem:time_of_S_dot_A}, the expected running time for calculating $SA$ is $O(\nnz(A) + \poly(d))$, and $x^* = (SA)^+ Sb$ can be calculated in $\poly(d)$ time.

To see the correctness, let $\overline{x} = \argmin_x \|Ax - b\|_{\ell}$.
With probability at least $0.99$, we have
\begin{align*}
\| Ax^*-b \|_{\ell} \leq & ~ O(\sqrt{d} \log^3 n) \| S Ax^* -S  b  \|_2 \\
	\leq & ~ O(\sqrt{d} \log^3 n) \| S A \overline{x} - S b \|_2 \\
	\leq & ~ O(\sqrt{d} \log^3 n)  \| \wt{D}A \overline{x} - \wt{D} b \|_{\ell} \\
	= & ~  O(\sqrt{d} \log^{11/2} n  ) \cdot \mmc(\ell) \cdot \| A \overline{x} - b \|_{\ell}.
	\end{align*}
	The first step follows by applying Lemma~\ref{lem:contraction_for_all_symsketch} on $\Ab$, where we use $\Ab  \in \mathbb{R}^{n \times (d + 1)}$ to denote a matrix whose first $d$ columns are $A$ and the last column is $b$.
	The second step follows from the fact that $x^* = \argmin_x \|SAx - Sb\|_2$.
	The third step follows by Definition~\ref{def:symmetric_sketch} and Corollary~\ref{cor:fast_gaussian_transform}.
	The last step follows by applying Lemma~\ref{lem:linear_regression_no_dilation_detailed} on $A \overline{x} - b$.

%% file: exp.tex
\section{Experiments}\label{sec:exp}
In this section, we perform experiments to validate the practicality of our methods.

\paragraph{Experiment Setup.} We compare the proposed algorithms with baseline algorithms on the U.S. 2000 Census Data containing $n = 5\times10^6$ rows and $d = 11$ columns and UCI YearPredictionMSD dataset which has $n = 515,345$ rows and $d = 90$ columns. 
All algorithms are implemented in Python 3.7. To solve the optimization problems induced by the regression problems and their sketched versions, we invoke the \texttt{minimize} function in \texttt{scipy.optimize}.
Each experiment is repeated for {\em 25 times}, and the mean of the loss function value is reported.
In all experiments, we vary the sampling size or embedding dimension from $5d$ to $20d$, and observe their effects on the quality of approximation.

\paragraph{Experiments on Orlicz Norm.}  We compare our algorithm in Section~\ref{sec:orlicz} with uniform sampling and the embedding in~\cite{alszz18}. We also calculate the optimal solution to verify the approximation ratio. We try Orlicz norms induced by two different $G$ functions: Huber with $c = 0.1$ and ``$\ell_1-\ell_2$''.
See Table~\ref{tab:M-estimators} for definitions. Our experimental results in Figure~\ref{fig:exp_orlicz} clearly demonstrate the practicality of our algorithm. In both datasets, our algorithm outperforms both baseline algorithms by a significant margin, and achieves the best accuracy in almost all settings. 
\begin{figure}
\centering
\begin{subfigure}[b]{0.49\textwidth}
\includegraphics[scale=0.4]{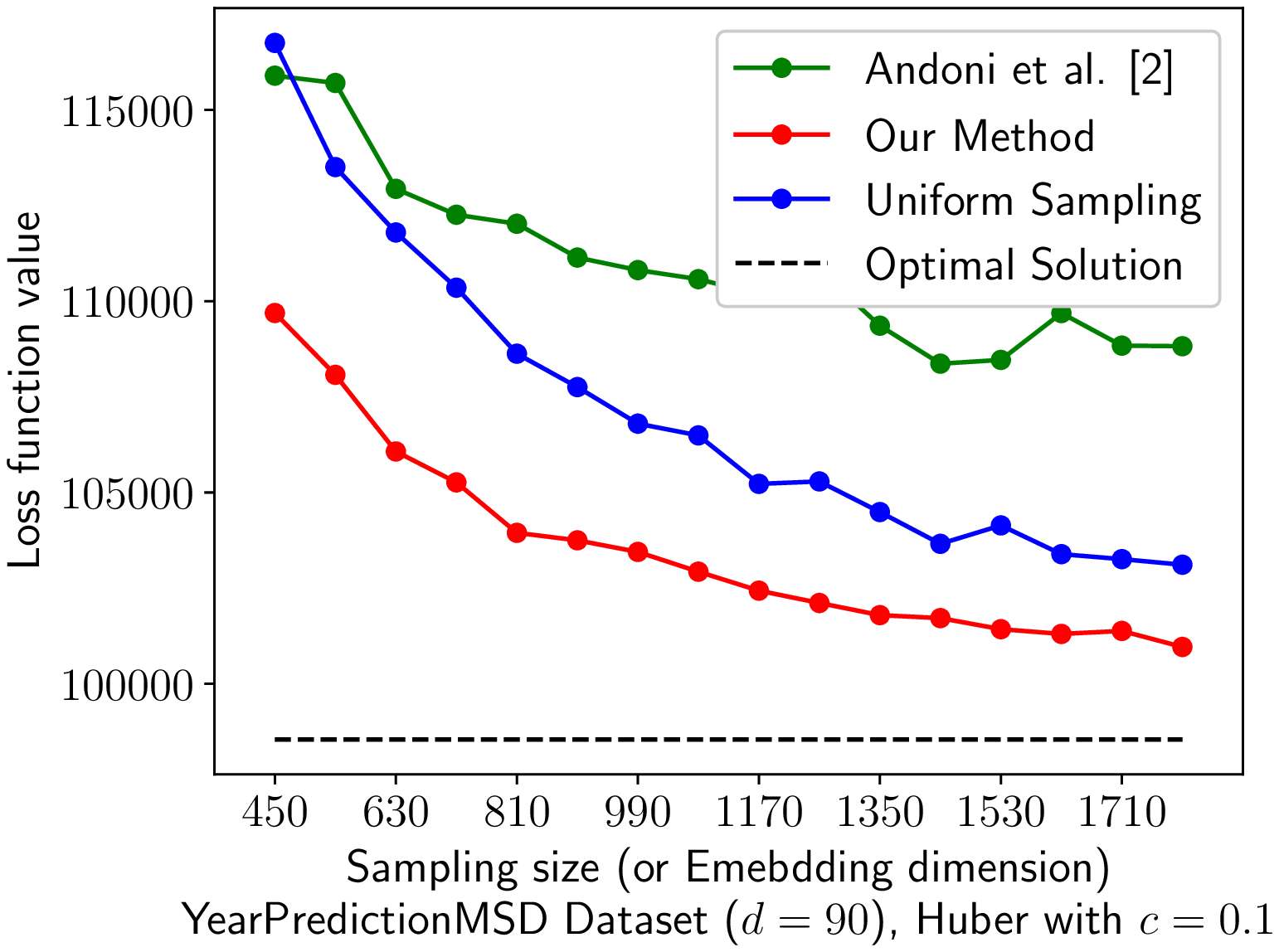}
\end{subfigure}
\begin{subfigure}[b]{0.49\textwidth}
\includegraphics[scale=0.4]{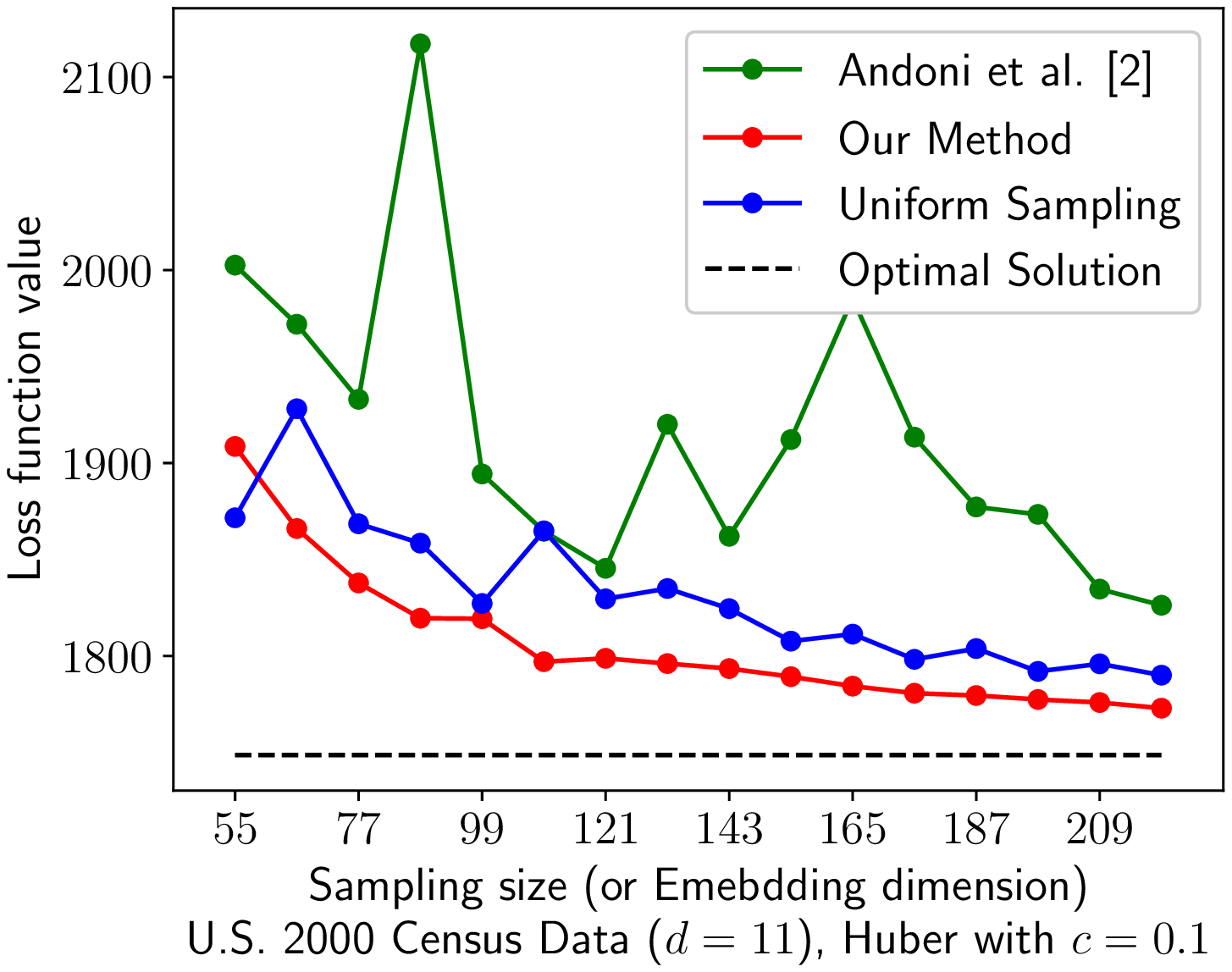}
\end{subfigure}
\begin{subfigure}[b]{0.49\textwidth}
\includegraphics[scale=0.4]{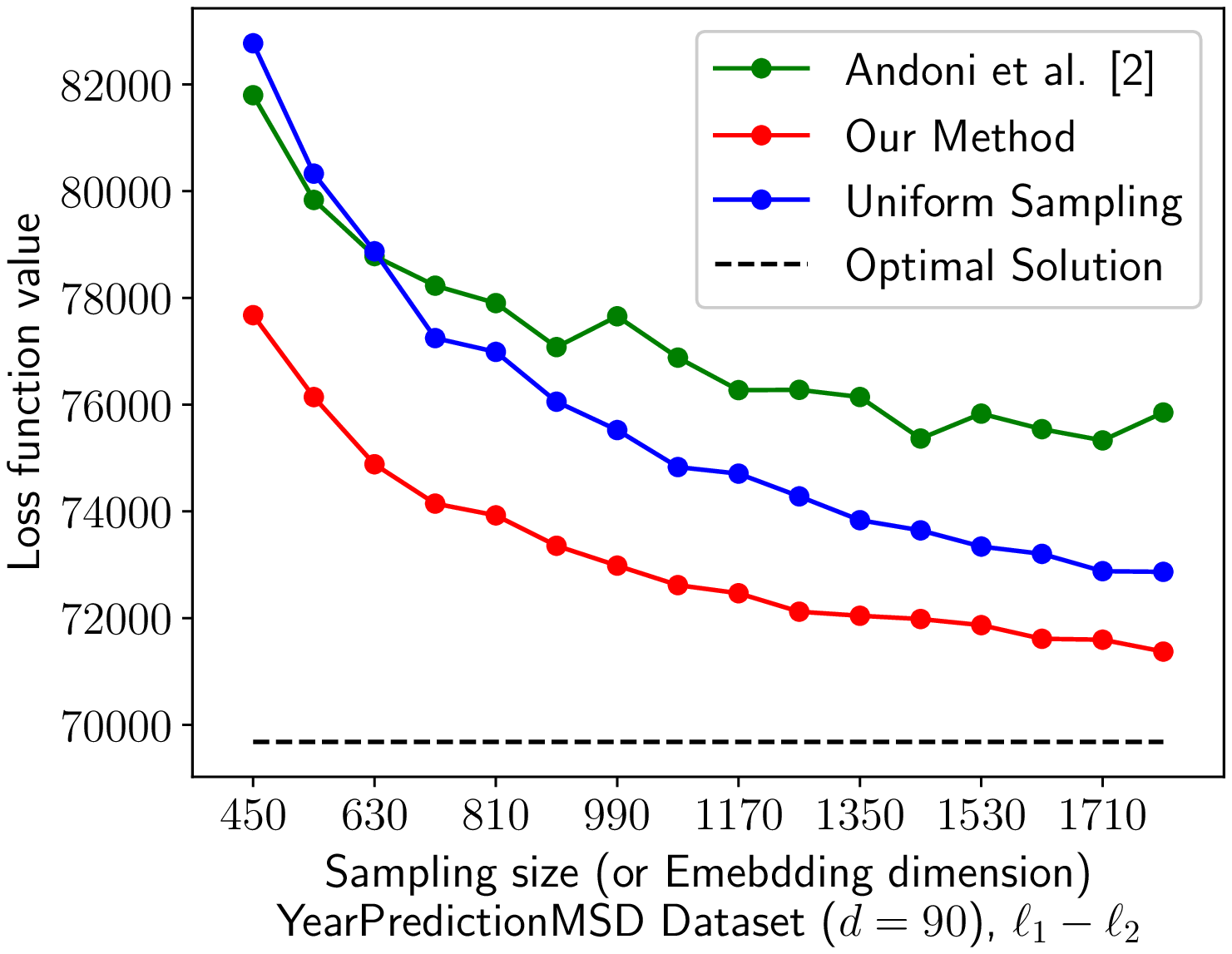}
\end{subfigure}
\begin{subfigure}[b]{0.49\textwidth}
\includegraphics[scale=0.4]{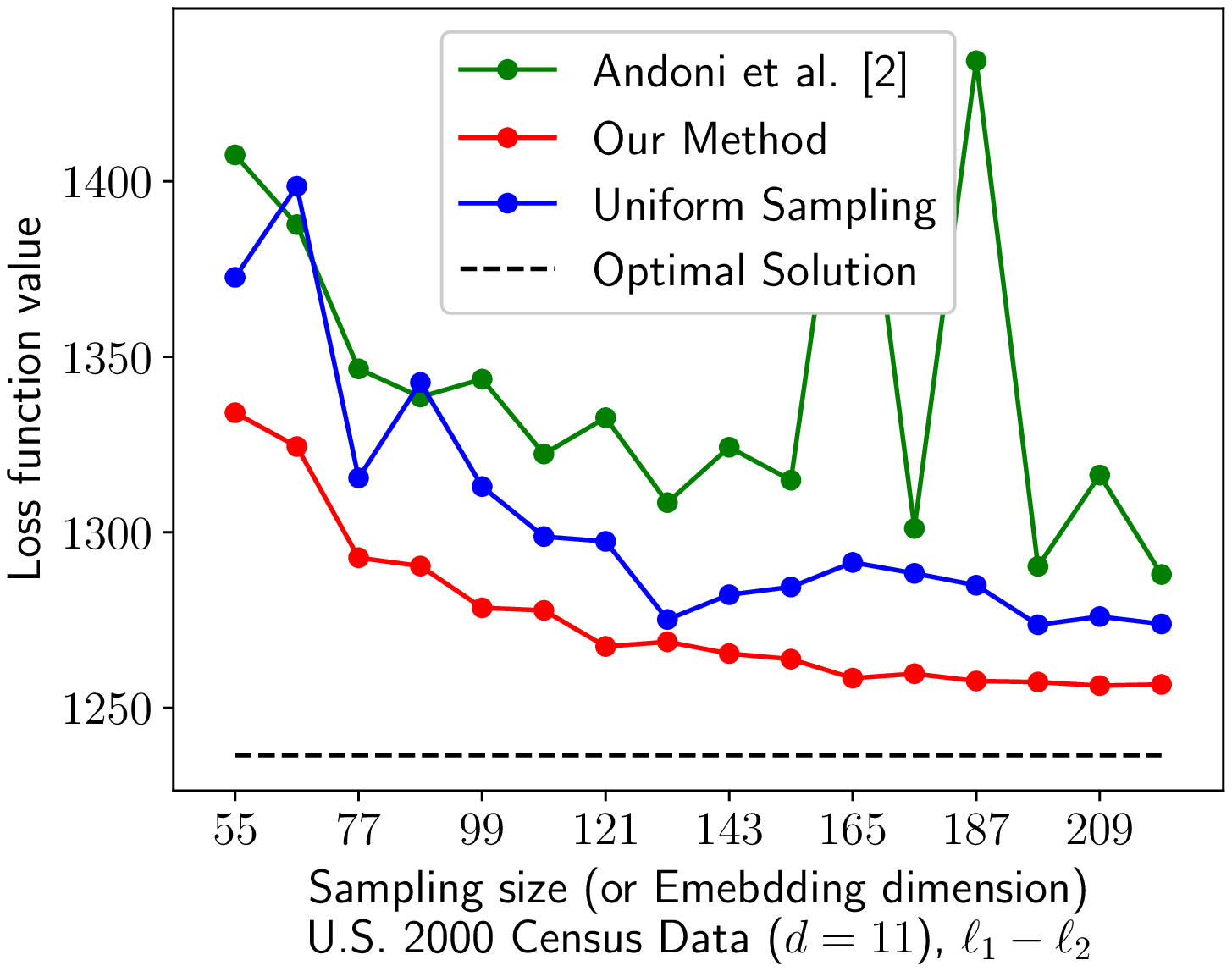}
\end{subfigure}
\caption{Experiments on Orlicz norm.} \label{fig:exp_orlicz}
\end{figure}

\paragraph{Experiments on Symmetric Norm.}  We compare our algorithm in Section~\ref{sec:alg} (\textsf{SymSketch}) with the optimal solution to verify the approximation ratio. We try two different symmetric norms: top-$k$ norm with $k = n / 5$ and sum-mix of $\ell_1$ and $\ell_2$ norm ($\|x\|_1 + \|x\|_2$). 
As shown in Figure~\ref{fig:exp_sym}, \textsf{SymSketch} achieves reasonable approximation ratios with moderate embedding dimension. In particular, the algorithm achieves an approximation ratio of $1.25$ when the embedding dimension is only $5d$.
\begin{figure}
\centering
\begin{subfigure}[b]{0.49\textwidth}
\includegraphics[scale=0.4]{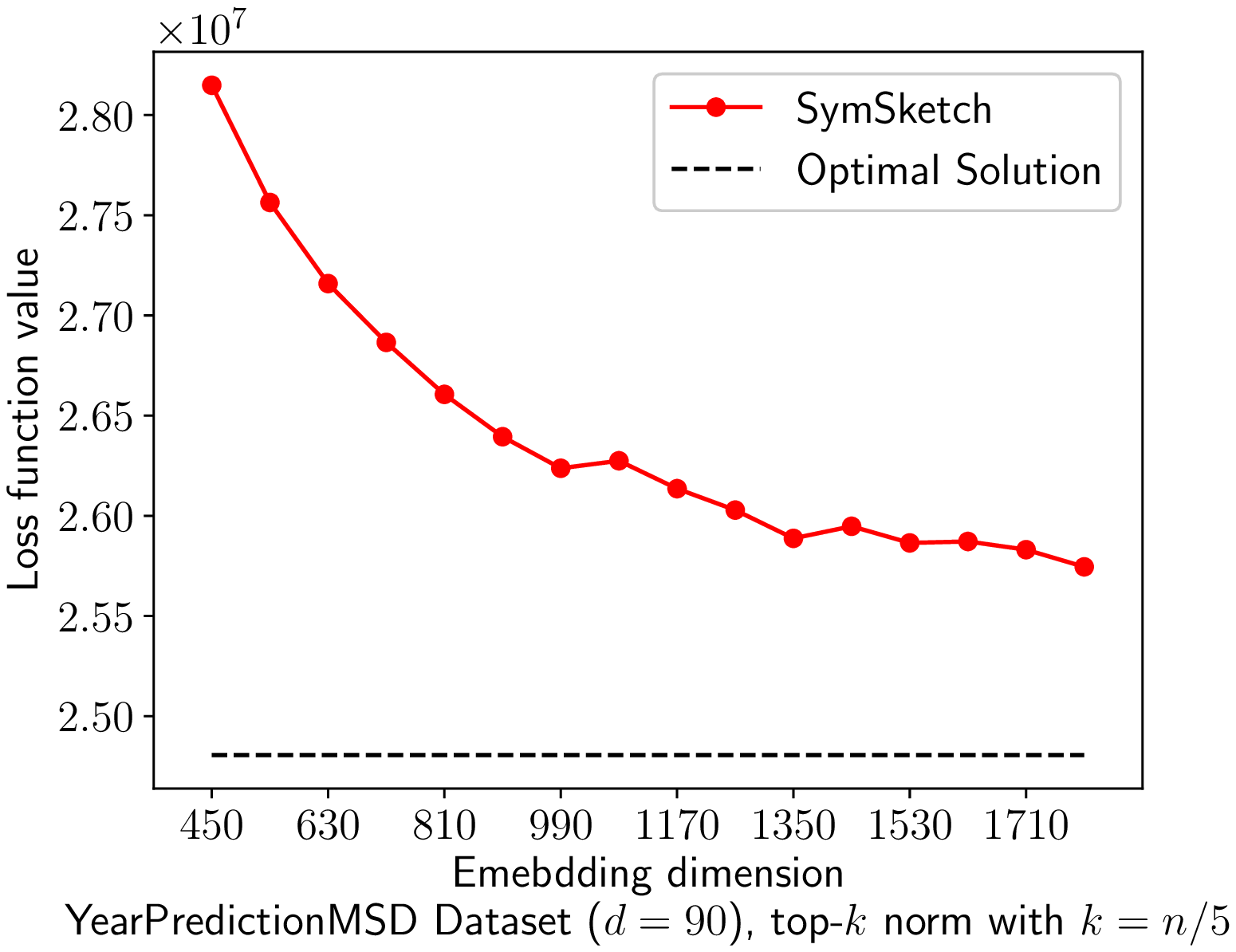}
\end{subfigure}
\begin{subfigure}[b]{0.49\textwidth}
\includegraphics[scale=0.4]{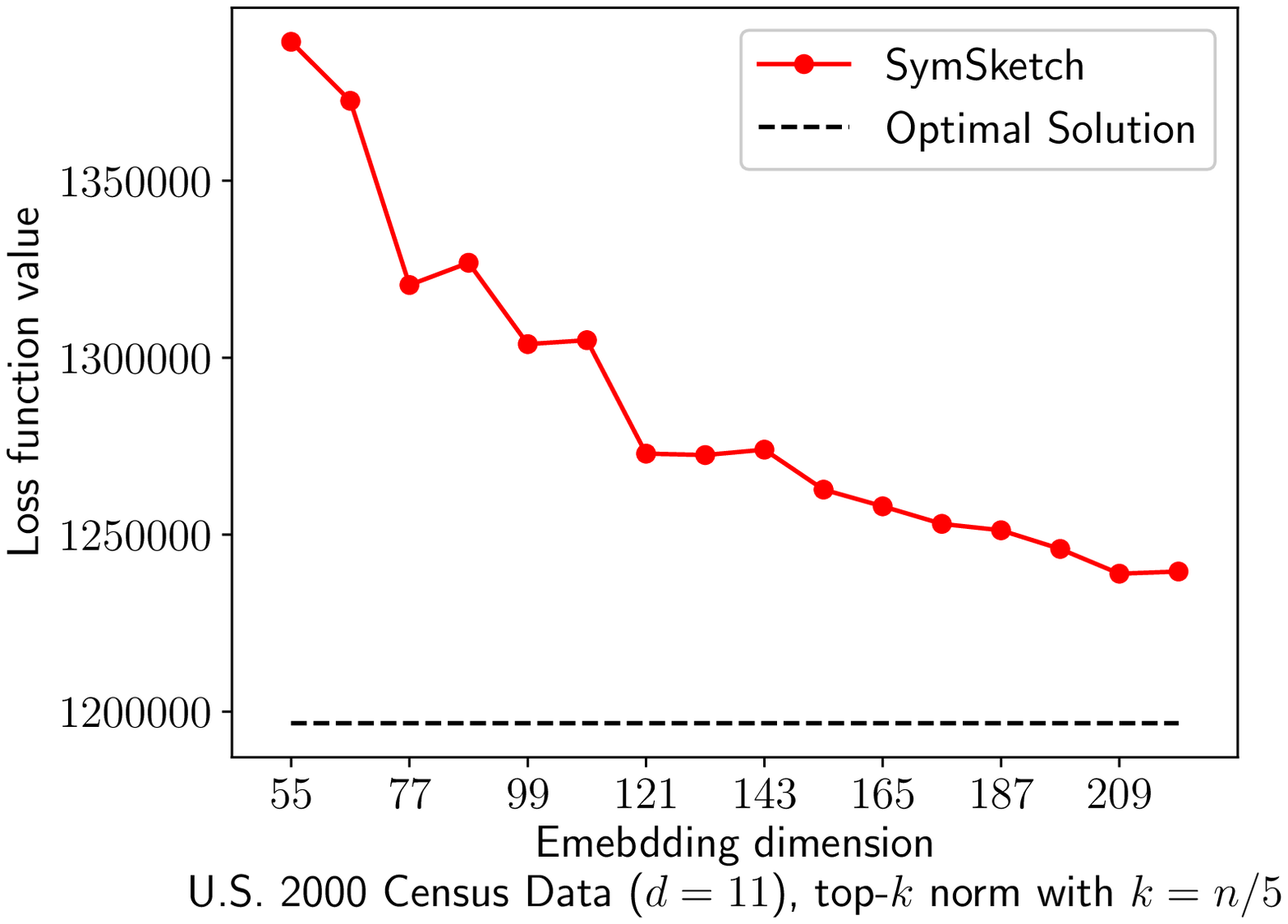}
\end{subfigure}
\begin{subfigure}[b]{0.49\textwidth}
\includegraphics[scale=0.4]{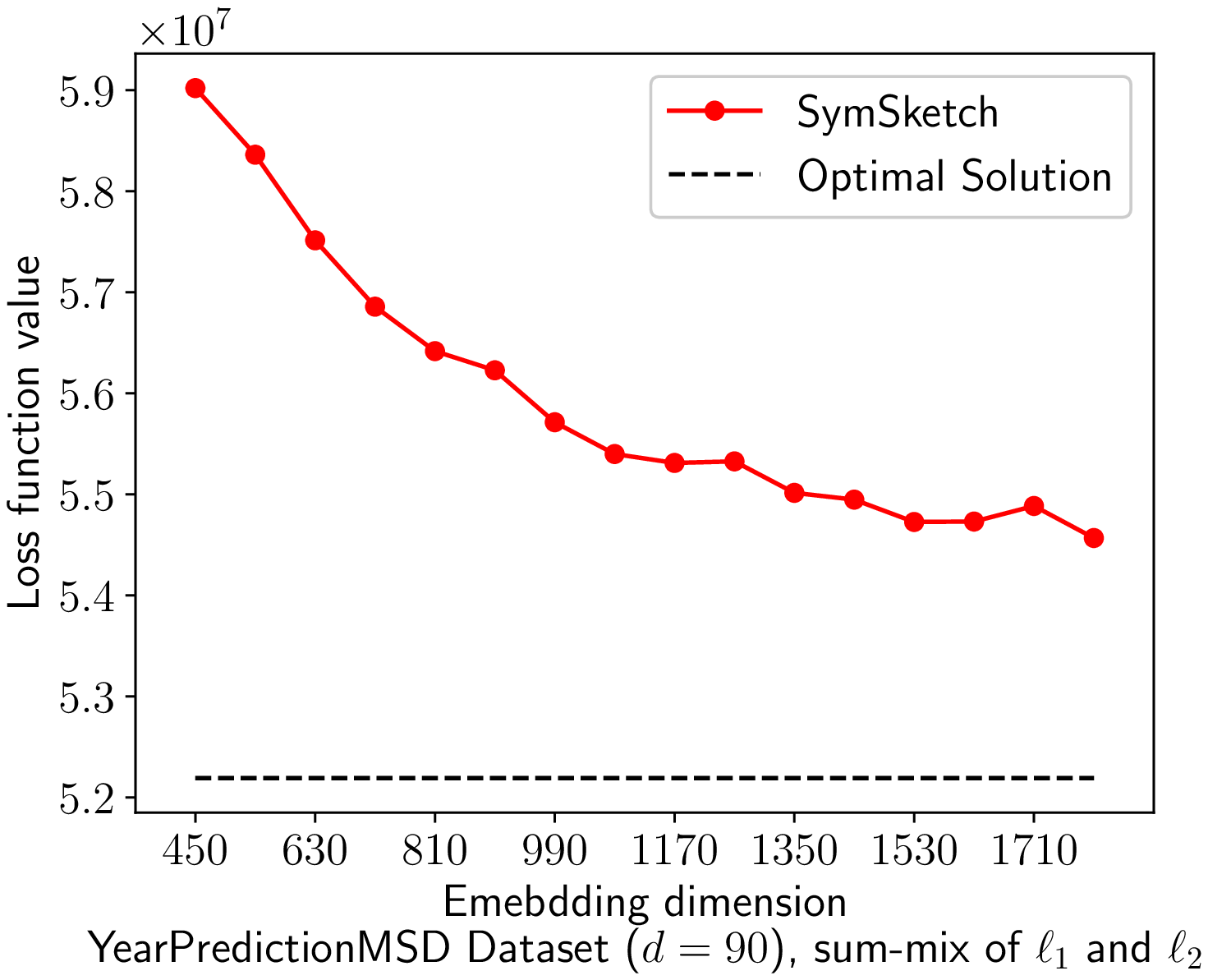}
\end{subfigure}
\begin{subfigure}[b]{0.49\textwidth}
\includegraphics[scale=0.4]{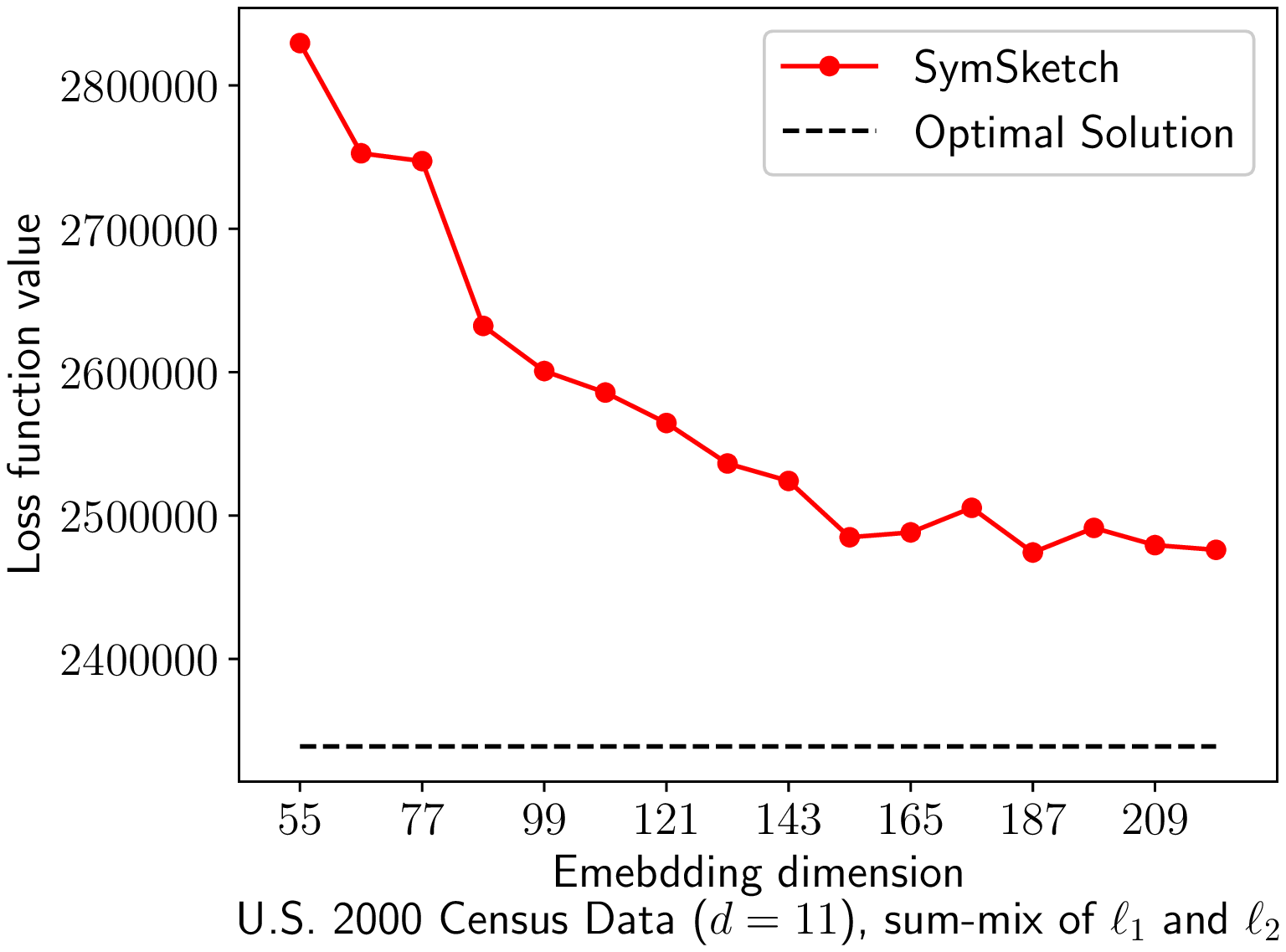}
\end{subfigure}
\caption{Experiments on symmetric norm.} \label{fig:exp_sym}
\end{figure}